\documentclass[fleqn,11pt,twoside]{article}

\usepackage{amsthm,amsthm,amssymb, color, xcolor,epsfig, graphics}
\usepackage{amsmath, graphicx, latexsym, lscape}

\usepackage{url, tikz, float}
\usepackage[colorlinks=true,urlcolor=black]{hyperref}

\makeatletter
\newcommand{\copyrightnote}[2]{{\renewcommand{\thefootnote}{}
 \footnotetext{\small\it
\begin{flushleft}
 \copyright \ #1   #2  
\end{flushleft}}}}

\newcommand{\Name}[1]{\begin{flushleft}
                       \LARGE \bf #1
                       \end{flushleft}\vspace{-3mm}}

\newcommand{\Author}[1]{\begin{flushleft}
                       \it #1 \end{flushleft}}

\newcommand{\Address}[1]{\begin{flushleft}
                       \it #1 \end{flushleft}}

\newcommand{\Date}[1]{\begin{flushleft}
                      \small  \it #1 \end{flushleft}}

\newcommand{\evenhead}{Author \ name}
\newcommand{\oddhead}{Article \ name}

\renewcommand{\@evenhead}{
\hspace*{-3pt}\raisebox{-15pt}[\headheight][0pt]{\vbox{\hbox to \textwidth
{\thepage \hfil \evenhead}\vskip4pt \hrule}}}
\renewcommand{\@oddhead}{
\hspace*{-3pt}\raisebox{-15pt}[\headheight][0pt]{\vbox{\hbox to \textwidth
{\oddhead \hfil \thepage}\vskip4pt\hrule}}}
\renewcommand{\@evenfoot}{}
\renewcommand{\@oddfoot}{}

\long\def\@makecaption#1#2{%
  \vskip\abovecaptionskip
  \sbox\@tempboxa{\small \textbf{#1.}\ \ #2}%
  \ifdim \wd\@tempboxa >\hsize
    {\small \textbf{#1.}\ \ #2}\par
  \else
    \global \@minipagefalse
    \hb@xt@\hsize{\hfil\box\@tempboxa\hfil}%
  \fi
  \vskip\belowcaptionskip}

\newcommand{\JNMPnumberwithin}[3][\arabic]{%
  \@ifundefined{c@#2}{\@nocounterr{#2}}{%
    \@ifundefined{c@#3}{\@nocnterr{#3}}{%
      \@addtoreset{#2}{#3}%
      \@xp\xdef\csname the#2\endcsname{%
        \@xp\@nx\csname the#3\endcsname .\@nx#1{#2}}}}%
}

\renewenvironment{proof}[1][\proofname]{\par
  \normalfont
  \topsep6\p@\@plus6\p@ \trivlist
  \item[\hskip\labelsep\textbf{%
    #1\@addpunct{.}}]\ignorespaces
}{%
  \qed\endtrivlist
}

\newcommand{\resetfootnoterule} {
  \renewcommand\footnoterule{%
  \kern-3\p@
  \hrule\@width.4\columnwidth
  \kern2.6\p@}
}

\renewcommand{\footnoterule}{}

\makeatother

\newcounter{NN}
\newtheorem{proposition}[NN]{Proposition}
\newtheorem{theorem}[NN]{Theorem}
\newtheorem{definition}[NN]{Definition}

\def\A{\mathbf{A}}
\def\B{\mathbf{B}}
\def\C{\mathcal{C}}
\def\Z{\mathbf{Z}}
\def\J{\mathbf{J}}
\def\K{\mathbf{K}}
\def\L{\mathbf{L}}
\def\M{\mathbf{M}}
\def\N{\mathbf{N}}
\def\b{\mathbf{r}}
\def\x{\mathbf{x}}
\def\1{\mathbf{1}}
\def\f{\mathbf{f}}
\def\k{\mathbf{k}}

\begin{document}

\renewcommand{\evenhead}{ {\LARGE\textcolor{blue!10!black!40!green}{{\sf \ \ \ ]ocnmp[}}}\strut\hfill 
P.H. van der Kamp, D.I. McLaren and G.R.W. Quispel
}
\renewcommand{\oddhead}{ {\LARGE\textcolor{blue!10!black!40!green}{{\sf ]ocnmp[}}}\ \ \ \ \  
Liouville integrable Lotka-Volterra systems
}

\thispagestyle{empty}
\newcommand{\FistPageHead}[3]{
\begin{flushleft}
\raisebox{8mm}[0pt][0pt]
{\footnotesize \sf
\parbox{150mm}{{\textcolor{blue!10!black!40!green}{{\bf Open Communications in Nonlinear Mathematical Physics}}}
\ \ {Special Issue: Hietarinta}, 2026\\[0.1cm]
\strut\hfill 
ocnmp:18125
pp #2\hfill {\sc #3}}}\vspace{-13mm}
\end{flushleft}}

\FistPageHead{1}{\pageref{firstpage}--\pageref{lastpage}}{ \ \ }

\strut\hfill

\strut\hfill

\copyrightnote{The authors. Distributed under a Creative Commons Attribution 4.0 International License}

\begin{center}

{\bf {\large A Special OCNMP Issue in Honour of Jarmo Hietarinta}}\\[0.2cm]
{\bf {\large on the Occasion of his 80th Birthday}}
\end{center}

\smallskip

\Name{Liouville integrable Lotka-Volterra systems}

\Author{Peter H.~van der Kamp\footnote{Email: P.vanderKamp@LaTrobe.edu.au}, David I. McLaren and G.R.W. Quispel}

\Address{Department of Mathematical and Physical Sciences,\\
La Trobe University, Victoria 3086, Australia.}

\Date{Received April 30, 2026; Accepted June 11, 2026}

\setcounter{equation}{0}

\smallskip

\noindent
{\bf Citation format for this Article:}\newline
P.H. van der Kamp, D.I. McLaren and G.R.W. Quispel,
Liouville integrable Lotka-Volterra systems,
{\it Open Commun. Nonlinear Math. Phys.}, Special Issue:\,Hietarinta, ocnmp:18125, \pageref{firstpage}--\pageref{lastpage}, 2026.

\strut\hfill

\noindent
{\bf The permanent Digital Object Identifier (DOI) for this Article:}\newline
{\it 10.46298/ocnmp.18125}
\strut\hfill

\begin{abstract}
\noindent 
We present $
\frac{m^{2}}{4}+\frac{m}{2}+\frac{1-\left(-1\right)^{m}}{8}$ homogeneous $(3m-2)$-parameter families of Liouville integrable $(2m)$- and $(2m-1)$-dimensional Lotka-Volterra systems. We also study inhomogeneous versions of these systems.
\end{abstract}

\label{firstpage}


\section{Introduction}
Jarmo Hietarinta has spent a large part of his career searching for new integrable systems \cite{JH1,JH2,JH3,JH4,JH5,JH6,JH7,JH8}, and one of his early interests has been Liouville integrability \cite{JH0}.
In this paper, we search for Liouville integrable Lotka-Volterra (LV) systems, i.e. ordinary differential equations of the form
\begin{equation} \label{LV}
\dot{x}_i=x_i\left(r_i+\sum_{j=1}^n A_{i,j}x_j\right),\quad i=1,2,\ldots,n,
\end{equation}
where the vector $\b$ and matrix $\A$ do not depend on $\x,t$, and $\dot{\x}$ denotes the $t$-derivative of $\x$. They generalise the famous 2-dimensional predator-prey system
\begin{equation} \label{PP}
\begin{split}
\dot{x}_1=x_1\left(r_1+A_{1,2}x_2\right)\\
\dot{x}_2=x_2\left(r_2+A_{2,1}x_1\right)
\end{split}
\end{equation}
derived independently by Lotka and by Volterra \cite{Lot1,Lot2,Lot3,Vol}. So, in essence, we will be extending integrable Hamiltonian LV systems from 2 to $n$ dimensions, cf. \cite{GDRH}.

As a matter of fact, Volterra himself already studied {\em the case of any number whatever of species some of which feed upon others} \cite[$\S$7]{Vol}, and, in \cite[$\S$5]{Vol2}, he presented us with an integrable $n$-species model. In this model, the growth coefficients $r_i$ are arbitrary parameters and the interaction matrix acquires the special form
\begin{equation} \label{VOLA}
A_{i,j}=r_ir_j(a_i-a_j),
\end{equation}
involving a second set of $n$ parameters $a_i$, $i=1,2,\ldots,n$. Volterra's system can be written as a Hamiltonian system $\dot{x_i}=\{x_i,H\}$, where
\[
H=\left(\sum_{i=1}^n x_i\right) + \frac{1}{a_1-a_2}\left(\frac{\ln(x_2)}{r_2}-\frac{\ln(x_1)}{r_1} \right)
\]
and the Poisson bracket is the log-canonical bracket \cite{FerOli,Plank}
\begin{equation} \label{PBA}
\{x_i,x_j\}=A_{i,j}x_ix_j.
\end{equation}
The rank of the matrix $\A$ given by \eqref{VOLA}, and hence of the corresponding bracket \eqref{PBA}, is 2, $n-2$ Casimir functions are given by
\[
x_{1}^{r_{2}r_{i}(a_{2} - a_{i})}
x_{2}^{r_{1}r_{i}(a_{i} - a_{1})}
x_{i}^{r_{1}r_{2}(a_{1} - a_{2})},\quad i=3,\ldots,n,
\]
and Liouville integrability coincides with maximally superintegrability, cf. \cite{RagZul,ScaRagTirZul}. A historical note should be made here, as Volterra did not use the Poisson bracket \eqref{PBA}. He introduced canonical variables on an extended $2n$ dimensional space, and showed (for a general anti-symmetric matrix $\A$) the existence of $n$ (time-dependent) integrals. He then showed that these integrals are in involution for the special rank 2 matrix \eqref{VOLA}. The relation between Volterra's canonical formalism and the log-canonical bracket \eqref{PBA} can be found in \cite{FerOli}. In the present work, we generalise the 2-dimensional rank 2 Hamiltonian LV system \eqref{PP}, cf. \cite{Nut}, to $n$-dimensional rank $n$ systems, when $n$ is even, and to $n$-dimensional rank $n-1$ systems when $n$ is odd.

LV systems \eqref{LV} have been studied extensively and they found many applications in e.g. biology and chemistry \cite{Gor,HadMacSte,Hof,Mur,Now,NM,Ste,Tak}. One of the interesting things about LV systems is that they can exhibit regular or chaotic behaviour, or both. Of course, systems with maximum regularity are integrable.

A Darboux Polynomial (DP) is a polynomial $P$ whose derivative, $\dot{P}$, is proportional to $P$. The entire $n(n+1)$-parameter class of LV systems \eqref{LV} admits $n$ such polynomials, namely the variables $x_i$, $i=1,\ldots,n$. By specialisation of the parameters, we obtain subclasses with additional DPs, which give rise to integrals, or other mathematical structures such as a Poisson bracket. Darboux polynomials may be used to divide the phase space in regions where the dynamics \cite{Gor,HadMacSte} present certain properties (bounded or unbounded motion for example).

Recently, the authors have established large classes of integrable cases. This includes:
\begin{itemize}
\item the study \cite{Poisson}, of inhomogeneous LV systems of the form
\begin{equation} \label{LVF}
\dot{x}_i=x_i\,\big(r_i + \sum_{j>i} x_j - \sum_{j<i}x_j\big),\qquad
r_i=\begin{cases}
b, & i=1,\ldots,k,\\
c, & i=k+1,\ldots, k+l,\\
d, & i=k+l+1,\ldots, k+l+m=n,
\end{cases}
\end{equation}
where $c=b+d$ and $km\neq 0$. The 2-parameter class \eqref{LVF} contains Liouville integrable, Liouville superintegrable, nonholonomically integrable, and 'integrable-or-not' subclasses. For the latter we were not able to find a Hamiltonian structure or more than $n-3$ integrals. Still the entire class was shown to be solvable, using the solutions to contracted 2- and 3-component systems given in terms of the Lambert W function.
\item the construction of LV tree-systems in \cite{Trees1,Trees2}. These
are homogeneous superintegrable $(3n-2)$-parameter classes of $n$-component LV systems. They admit $n-1$ functionally independent integrals that involve $n-1$ linear DPs in two variables (2-DPs), which correspond to the edges of a tree. We note that these systems are not Hamiltonian. A very special case (1-parameter) is the generalised Kovaleskaya system \cite{PS}, with $\b={\bf 0}$ and 
\[
A_{i,j}=\begin{cases} \alpha & i=j \\ 1 & i\neq j. \end{cases}
\]
For this system, $n-1$ functionally independent integrals are given by
\[
\left(\frac{x_i - x_{j}}{x_ix_{j}}\right)^{n + \alpha - 1}\prod_{k=1}^n x_k, \quad j=i+1,i=1,2,\ldots,n-1.
\]
Its graph (in the sense of 2-DPs \cite{MP}) is the complete graph on $n$ vertices.

\item the classification of classes of $n$-component LV systems with linear DPs in any number of variables, for $n<6$, up to linear transformations \cite{Hyper}. These systems are associated to (LV admissible) hypergraphs. One new superintegrable $3n-2=13$ parameter family of $n=5$ component LV systems was found, not equivalent to a tree system.
\end{itemize}

In this paper, we assume that $\A$ is an anti-symmetric constant $n \times n$ matrix, starting with $n=2m$. In the homogeneous case, this implies that the LV system is a Hamiltonian system, with Hamiltonian
\begin{equation} \label{H}
H=x_1+x_2+\cdots+x_n.
\end{equation}
and Poisson bracket \eqref{PBA}. In other words, the equation
\begin{equation} \label{HF}
\dot{x}_i=\{x_i,H\},\quad i=1,\ldots,n,
\end{equation}
equals \eqref{LV} with $\b=0$. We follow a similar approach to the one taken in \cite{Trees1,Hyper,Trees2}, i.e., we require that certain linear functions are DPs of the system (each of them giving rise to an integral). Note that we do not need as many as before. Whereas for superintegrability $n-1$ were needed, here, besides the Hamiltonian, $n/2-1$ DPs suffice. On the other hand, we do require the corresponding integrals to Poisson commute. Thus we establish Liouville integrability which is defined as follows \cite{LI}.
\begin{definition}
An $n$-dimensional Hamiltonian system is Liouville integrable if there are $s$ Casimir functions and $\frac{n-s}{2}$ commuting and functionally independent integrals (including the Hamiltonian).
\end{definition}

We present a large number of homogeneous and inhomogeneous Liouville integrable multi-parameter families of $n$-dimensional Lotka-Volterra systems, for each $n\in\N$. In Section \ref{Seven}, we present
\[
\frac{m^{2}}{4}+\frac{m}{2}+\frac{1-\left(-1\right)^{m}}{8}
\]
homogeneous $n=2m$ dimensional $3m-2$ parameter families of LV systems, which have $m$ integrals in involution. In Section \ref{Sodd}, we present equally many odd $n=2m-1$ dimensional LV systems, for which we have $m-1$ integrals in involution and 1 Casimir. In Section \ref{Snonh}, we show that each even-dimensional system has a subsystem to which one can add a inhomogeneous term of the form
\[
\b=(r,\ldots,r,s).
\]
In odd dimensions, this is not the case. In Section \ref{explex}, we provide explicit details for the families of Liouville integrable LV systems we have found in dimensions 3,4,5,6, as well as for some Liouville integrable LV systems of different type in dimension $n=10$. We give concluding remarks in Section \ref{ConclRemks}.

\section{Homogeneous Liouville integrable LV systems} \label{Shom}
In this section, we consider homogeneous $n$-dimensional Hamiltonian systems \eqref{HF}, with anti-symmetric matrix $\A$ and Hamiltonian \eqref{H}.
Integrals for \eqref{HF} are functions $F$ for which $\dot{F}=\{F,H\}=0$. A {\bf Darboux Polynomial} (DP) is a polynomial $P$ for which
\begin{equation} \label{CP}
\dot{P}=\{P,H\}=CP,
\end{equation}
in which case $C$ is called the cofactor of $P$, also denoted $C[P]$. Clearly, the Hamiltonian itself is an integral, i.e., a DP with cofactor 0. More generally, a {\em function} $P$ which satisfies \eqref{CP} is called a {\bf second integral} \cite[Definition 2.14]{Gor}. A product of powers of DPs,
\begin{equation}
\Pi=\prod_i P_i^{\alpha_i}
\end{equation}
is a second integral with cofactor
\[
C[\Pi]=\sum_i\alpha_iC[P_i].
\]

We will require linear combinations of the form
\begin{equation} \label{DP}
P_I=\alpha_1 x_{i_1} + \alpha_2 x_{i_2} + \cdots + \alpha_k x_{i_k},
\end{equation}
where
\[
I=\{i_1,\ldots,i_k\}\subset\mathbb{N}_n=\{x_1,\ldots,x_n\},
\]
to be DPs. This imposes conditions on the matrix $\A$, as given in \cite{Trees1} for $k=2$, or in \cite{Hyper} for $k>2$. Due to the antisymmetry of $\A$, the coefficients $\alpha_i$ in $P_I$ 
are all equal; we take $\alpha_i=1$ for all $i$. If the matrix $\A$ has full rank, then each DP gives rise to a integral. Indeed, if $\B$ is the $1\times n$ vector such that
\[
C[P_I]=\sum_{i=1}^n B_ix_i,
\]
and $\Z=-\B.\A^{-1}$, then
\begin{equation} \label{integral}
P_I\prod_{i=1}^n x_i^{Z_{i}}
\end{equation}
is an integral, cf. \cite[Section 2]{Trees1}. Further conditions on the matrix $\A$ are imposed by requiring the integrals to Poisson commute. A homogeneous $n$-dimensional Hamiltonian LV system with $k$ integrals of the form \eqref{integral} which are pairwise commuting will be represented by a hypergraph of order $n$ and size $k$ (i.e. a subset, with $k$ elements, of the powerset of $n$ symbols).

\bigskip
For example, the Hamiltonian LV system with matrix
\begin{equation} \label{fa}
\A=\begin{pmatrix}
0 & a_{1} & b_{1} & b_{1} & b_{1} & c_{1} 
\\
 -a_{1} & 0 & b_{1} & b_{1} & b_{1} & c_{1} 
\\
 -b_{1} & -b_{1} & 0 & a_{2} & a_{3} & a_{5} 
\\
 -b_{1} & -b_{1} & -a_{2} & 0 & a_{4} & a_{5} 
\\
 -b_{1} & -b_{1} & -a_{3} & -a_{4} & 0 & a_{5} 
\\
 -c_{1} & -c_{1} & -a_{5} & -a_{5} & -a_{5} & 0 
\end{pmatrix}
\end{equation}
has a 2-DP,
$
x_1+x_2$,
and a 3-DP,
$
x_3+x_4+x_5$.
The corresponding integrals are
\[
\left(x_{1}+x_{2}\right)
\left( 
\frac{x_{4}^{a_{3}}}{x_{3}^{a_{4}}x_{5}^{a_{2}}}
\right)^{\frac{c_{1}}{a_{5} \left(a_{2}-a_{3}+a_{4}\right)}} x_{6}^{\frac{b_{1}}{a_{5}}},\quad 
\left(x_{3}+x_{4}+x_{5}\right) \left( 
\frac{x_{4}^{a_{3}}}{x_{3}^{a_{4}}x_{5}^{a_{2}}}
\right)^{\frac{1}{a_{2}-a_{3}+a_{4}}},
\]
and they Poisson commute. This is represented by the hypergraph in Figure \ref{Fe}.

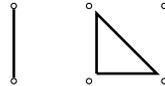
\begin{figure}[h!]
\centering
\begin{tikzpicture}[scale=1]
\tikzstyle{nod}= [circle, inner sep=0pt, fill=white, minimum size=2pt, draw]		
\node[nod] (a) at (0,0) {};
\node[nod] (b) at (0,1) {};
\node[nod] (c) at (1,0) {};
\node[nod] (d) at (1,1) {};
\node[nod] (e) at (2,0) {};
\node[nod] (f) at (2,1) {};
\coordinate (cs) at (1.1,0.1) {};
\coordinate (ds) at (1.1,0.9) {};
\coordinate (es) at (1.9,0.1) {};
\draw[line width=1pt] (a) -- (b);
\draw[line width=1pt] (cs) -- (ds) -- (es) -- (cs);
\end{tikzpicture}
\caption{\label{Fe} The hypergraph on $6$ vertices of size 2, associated to the LV system with matrix \eqref{fa}.}
\end{figure}

We note that this representation differs from the one in \cite{Hyper}, in more than one way. First of all, the matrices in \cite{Hyper} are not, in general, antisymmetric. In the current paper, the matrix $\A$ is assumed to be anti-symmetric which implies that, in all cases, the function $H$ is a DP (and an integral), but we omit the corresponding hyperedge (involving all vertices). Secondly, as stated above, here we require that the specified integrals, obtained from \eqref{integral}, Poisson commute. 

\begin{figure}[H]
\centering
\begin{tikzpicture}[scale=1]
\tikzstyle{nod}= [circle, inner sep=0pt, fill=white, minimum size=2pt, draw]		
\node[nod] (a) at (0,0) {};
\node[nod] (b) at (0,1) {};
\node[nod] (c) at (1,0) {};
\node[nod] (d) at (1,1) {};
\node[nod] (e) at (3,0) {};
\node[nod] (f) at (3,1) {};
\node[nod] (g) at (4,0) {};
\node[nod] (h) at (4,1) {};
\node at (2,1/2) {$\ldots$};
\draw[line width=1pt] (a) -- (b);
\draw[line width=1pt] (c) -- (d);
\draw[line width=1pt] (e) -- (f);
\end{tikzpicture}
\\[8mm]
\begin{tikzpicture}[scale=1]
\tikzstyle{nod}= [circle, inner sep=0pt, fill=white, minimum size=2pt, draw]		
\node[nod] (a) at (0,0) {};
\node[nod] (b) at (0,1) {};
\node[nod] (c) at (1,0) {};
\node[nod] (d) at (1,1) {};
\node[nod] (e) at (3,0) {};
\node[nod] (f) at (3,1) {};
\node[nod] (g) at (4,0) {};
\node[nod] (h) at (4,1) {};
\node[nod] (i) at (5,0) {};
\node[nod] (j) at (5,1) {};
\coordinate (js) at (4.9,0.9) {};
\coordinate (hs) at (4.1,0.9) {};
\coordinate (is) at (4.9,0.1) {};
\node at (2,1/2) {$\ldots$};
\draw[line width=1pt] (a) -- (b);
\draw[line width=1pt] (c) -- (d);
\draw[line width=1pt] (e) -- (f);
\draw[line width=1pt] (js) -- (hs) -- (is) -- (js);
\end{tikzpicture}
\\[8mm]
\begin{tikzpicture}[scale=1]
\tikzstyle{nod}= [circle, inner sep=0pt, fill=white, minimum size=2pt, draw]		
\node[nod] (a) at (0,0) {};
\node[nod] (b) at (0,1) {};
\node[nod] (c) at (1,0) {};
\node[nod] (d) at (1,1) {};
\node[nod] (e) at (3,0) {};
\node[nod] (f) at (3,1) {};
\node[nod] (g) at (4,0) {};
\node[nod] (h) at (4,1) {};
\node[nod] (i) at (5,0) {};
\node[nod] (j) at (5,1) {};
\node[nod] (k) at (6,0) {};
\node[nod] (l) at (6,1) {};
\coordinate (js) at (4.9,0.9) {};
\coordinate (hs) at (4.1,0.9) {};
\coordinate (is) at (4.9,0.1) {};
\node at (2,1/2) {$\ldots$};
\draw[line width=1pt] (a) -- (b);
\draw[line width=1pt] (c) -- (d);
\draw[line width=1pt] (e) -- (f);
\draw[line width=1pt] (js) -- (hs) -- (is) -- (js);
\coordinate (ht) at (3.75,1.1) {};
\coordinate (it) at (4.95,-0.1) {};
\coordinate (kt) at (6.1,-0.1) {};
\coordinate (lt) at (6.1,1.1) {};
\draw[line width=1pt] (ht) -- (lt) -- (kt) -- (it) -- (ht);
\end{tikzpicture}
\\[8mm]
$\vdots$
\\[8mm]
\begin{tikzpicture}[scale=1]
\tikzstyle{nod}= [circle, inner sep=0pt, fill=white, minimum size=2pt, draw]		
\node[nod] (a) at (0,0) {};
\node[nod] (b) at (0,1) {};
\node[nod] (c) at (1,0) {};
\node[nod] (d) at (1,1) {};
\node[nod] (e) at (3,0) {};
\node[nod] (f) at (3,1) {};
\node[nod] (g) at (4,0) {};
\node[nod] (h) at (4,1) {};
\node[nod] (i) at (5,0) {};
\node[nod] (j) at (5,1) {};
\node[nod] (k) at (6,0) {};
\node[nod] (l) at (6,1) {};
\node[nod] (m) at (8,0) {};
\node[nod] (n) at (8,1) {};
\coordinate (js) at (4.9,0.9) {};
\coordinate (hs) at (4.1,0.9) {};
\coordinate (is) at (4.9,0.1) {};
\node at (2,1/2) {$\ldots$};
\node at (7,1/2) {$\ldots$};
\draw[line width=1pt] (a) -- (b);
\draw[line width=1pt] (c) -- (d);
\draw[line width=1pt] (e) -- (f);
\draw[line width=1pt] (js) -- (hs) -- (is) -- (js);
\coordinate (ht) at (3.75,1.1) {};
\coordinate (it) at (4.95,-0.1) {};
\coordinate (kt) at (6.1,-0.1) {};
\coordinate (lt) at (6.1,1.1) {};
\draw[line width=1pt] (ht) -- (lt) -- (kt) -- (it) -- (ht);
\coordinate (hr) at (3.5,1.2) {};
\coordinate (ir) at (4.9,-0.2) {};
\coordinate (kr) at (8.1,-0.2) {};
\coordinate (lr) at (8.1,1.2) {};
\draw[line width=1pt] (hr) -- (lr) -- (kr) -- (ir) -- (hr);
\end{tikzpicture}
\\[8mm]
\begin{tikzpicture}[scale=1]
\tikzstyle{nod}= [circle, inner sep=0pt, fill=white, minimum size=2pt, draw]		
\node[nod] (a) at (-1,0) {};
\node[nod] (b) at (-1,1) {};
\node[nod] (c) at (0,0) {};
\node[nod] (d) at (0,1) {};
\node[nod] (e) at (2,0) {};
\node[nod] (f) at (2,1) {};
\node[nod] (ex) at (3,0) {};
\node[nod] (fx) at (3,1) {};
\node[nod] (g) at (4,0) {};
\node[nod] (h) at (4,1) {};
\node[nod] (i) at (5,0) {};
\node[nod] (j) at (5,1) {};
\node[nod] (k) at (6,0) {};
\node[nod] (l) at (6,1) {};
\node[nod] (m) at (8,0) {};
\node[nod] (n) at (8,1) {};
\coordinate (js) at (4.9,0.9) {};
\coordinate (hs) at (4.1,0.9) {};
\coordinate (is) at (4.9,0.1) {};
\node at (1,1/2) {$\ldots$};
\node at (7,1/2) {$\ldots$};
\draw[line width=1pt] (a) -- (b);
\draw[line width=1pt] (c) -- (d);
\draw[line width=1pt] (e) -- (f);
\draw[line width=1pt] (js) -- (hs) -- (is) -- (js);
\coordinate (exp) at (3.1,0.1) {};
\coordinate (fxp) at (3.1,.9) {};
\coordinate (gxp) at (3.9,0.1) {};
\draw[line width=1pt] (exp) -- (fxp) -- (gxp) -- (exp);
\coordinate (ht) at (3.75,1.1) {};
\coordinate (it) at (4.95,-0.1) {};
\coordinate (kt) at (6.1,-0.1) {};
\coordinate (lt) at (6.1,1.1) {};
\draw[line width=1pt] (ht) -- (lt) -- (kt) -- (it) -- (ht);
\coordinate (hr) at (3.5,1.2) {};
\coordinate (ir) at (4.9,-0.2) {};
\coordinate (kr) at (8.1,-0.2) {};
\coordinate (lr) at (8.1,1.2) {};
\draw[line width=1pt] (hr) -- (lr) -- (kr) -- (ir) -- (hr);
\end{tikzpicture}
\\[8mm]
$\vdots$
\\[8mm]
\begin{tikzpicture}[scale=1]
\tikzstyle{nod}= [circle, inner sep=0pt, fill=white, minimum size=2pt, draw]		
\node[nod] (a) at (0,0) {};
\node[nod] (b) at (0,1) {};
\node[nod] (c) at (1,0) {};
\node[nod] (d) at (1,1) {};
\node[nod] (e) at (3,0) {};
\node[nod] (f) at (3,1) {};
\node[nod] (a1) at (4,0) {};
\node[nod] (b1) at (4,1) {};
\node[nod] (c1) at (5,0) {};
\node[nod] (d1) at (5,1) {};
\node[nod] (e1) at (6,0) {};
\node[nod] (h1) at (7,1) {};
\coordinate (a2) at (4.1,0.1) {};
\coordinate (b2) at (4.1,0.9) {};
\coordinate (c2) at (4.9,0.1) {};
\draw[line width=1pt] (a2) -- (b2) -- (c2) -- (a2);
\coordinate (a3) at (3.9,-0.1) {};
\coordinate (b3) at (3.9,1.1) {};
\coordinate (d3) at (5.05,1.1) {};
\coordinate (e3) at (6.25,-0.1) {};
\draw[line width=1pt] (a3) -- (b3) -- (d3) -- (e3) -- (a3);
\coordinate (a4) at (3.8,-0.2) {};
\coordinate (b4) at (3.8,1.2) {};
\coordinate (d4) at (7,1.2) {};
\coordinate (e4) at (8.4,-0.2) {};
\draw[line width=1pt] (a4) -- (b4) -- (d4) -- (e4) -- (a4);
\node[nod] (g) at (8,0) {};
\node[nod] (h) at (8,1) {};
\node[nod] (i) at (9,0) {};
\node[nod] (j) at (9,1) {};
\node[nod] (k) at (10,0) {};
\node[nod] (l) at (10,1) {};
\node[nod] (m) at (12,0) {};
\node[nod] (n) at (12,1) {};
\coordinate (js) at (8.9,0.9) {};
\coordinate (hs) at (8.1,0.9) {};
\coordinate (is) at (8.9,0.1) {};
\node at (2,1/2) {$\ldots$};
\node at (11,1/2) {$\ldots$};
\node at (6.7,1/2) {$\ldots$};
\draw[line width=1pt] (a) -- (b);
\draw[line width=1pt] (c) -- (d);
\draw[line width=1pt] (e) -- (f);
\draw[line width=1pt] (js) -- (hs) -- (is) -- (js);
\coordinate (ht) at (7.75,1.1) {};
\coordinate (it) at (8.95,-0.1) {};
\coordinate (kt) at (10.1,-0.1) {};
\coordinate (lt) at (10.1,1.1) {};
\draw[line width=1pt] (ht) -- (lt) -- (kt) -- (it) -- (ht);
\coordinate (hr) at (7.5,1.2) {};
\coordinate (ir) at (8.9,-0.2) {};
\coordinate (kr) at (12.1,-0.2) {};
\coordinate (lr) at (12.1,1.2) {};
\draw[line width=1pt] (hr) -- (lr) -- (kr) -- (ir) -- (hr);
\end{tikzpicture}
\\[5mm]
\caption{\label{F0} A tower of hypergraphs. Each hypergraph has a number of edges of degree 2 and at most two sets of nested hyperedges of odd degrees $3,5,\ldots$.}
\end{figure}
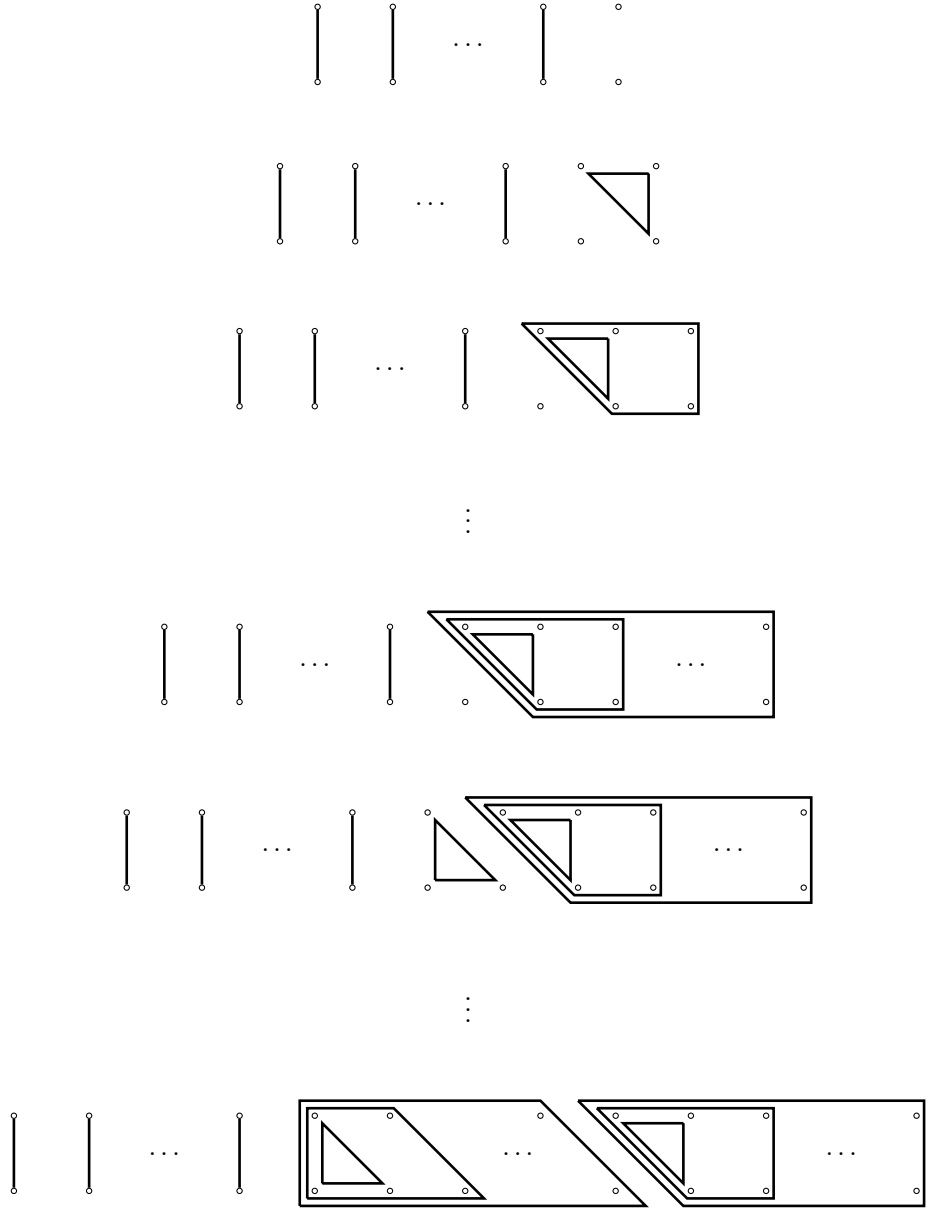

In Section \ref{Seven}, we consider even-dimensional LV systems. In each case, the matrix $\A$ will have full rank, so integrals can be calculated using the formula \eqref{integral}. We do however present different integrals, that are (or can be made) homogeneous of weight zero. This allows us to extend the systems to the inhomogeneous case, which we do in Section \ref{Snonh}. In Section \ref{Sodd} we show that each even-dimensional system can be reduced to an odd-dimensional system, whilst retaining its Liouville integrability. In the $n=2m$ dimensional case, we have $m$ functionally independent integrals (including the Hamiltonian) in involution. In the $n=2m-1$ dimensional case, there are $m-1$ functionally independent integrals (including the Hamiltonian) in involution, as well as a Casimir. 
    
\subsection{Even-dimensional homogeneous Liouville integrable LV systems} \label{Seven}
We will show there exists a $2m$-dimensional homogeneous Liouville integrable LV system with $3m-2$ parameters associated to each hypergraph of order $2m$  displayed in Fig. \ref{F0}.
These hypergraphs will be  labeled by the number of degree 2 edges and the two (ordered) numbers of nested edges. For example, the hypergraph in Fig. \ref{Fe} is labeled by $[1,1,0]$, and the hypergraphs displayed explicitly in Fig. \ref{F0} are labeled by $[m-1,0,0]$, $[m-2,1,0]$, $[m-3,2,0]$, $[m-k-1,k,0]$, $[m-k-2,k,1]$, and $[m-k-l-1,k,l]$ respectively (with $k\geq l\geq 0$ and $m>k+l$).


\subsubsection{The even-dimensional homogeneous type $[j,0,0]$ LV system}
The type $[j,0,0]$ LV system corresponds to the graph in Fig. \ref{F1}.

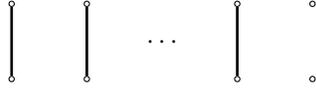
\begin{figure}[H]
\centering
\begin{tikzpicture}[scale=1]
\tikzstyle{nod}= [circle, inner sep=0pt, fill=white, minimum size=2pt, draw]		
\node[nod] (a) at (0,0) {};
\node[nod] (b) at (0,1) {};
\node[nod] (c) at (1,0) {};
\node[nod] (d) at (1,1) {};
\node[nod] (e) at (3,0) {};
\node[nod] (f) at (3,1) {};
\node[nod] (g) at (4,0) {};
\node[nod] (h) at (4,1) {};
\node at (2,1/2) {$\ldots$};
\draw[line width=1pt] (a) -- (b);
\draw[line width=1pt] (c) -- (d);
\draw[line width=1pt] (e) -- (f);
\end{tikzpicture}
\caption{\label{F1} A forest on $2(j+1)$ vertices.}
\end{figure}

We denote by $\1^{i,k}$ the $i\times k$ matrix with entries $1$. Using $2\times 2$ matrices
\begin{equation} \label{JKLehj00}
\J_i=\begin{pmatrix}
0 & a_i \\
-a_i & 0
\end{pmatrix},\
\K^i_k=\K^i_k(h)=\frac{b_ic_k-b_kc_i}{a_{h}}\1^{2,2},
\end{equation}
we define the $2j\times2j$ matrix
\begin{equation} \label{Mjh}
\M^j_h=\begin{pmatrix}
\J_1 & \K^1_2 & \K^1_3 & \cdots & \K^1_j\\
-\K^1_2 & \J_2 & \K^2_3 & \cdots & \K^2_j\\
-\K^1_3 & - \K^2_3 & \J_3 & \cdots & \K^3_j\\
\vdots & \vdots & \vdots & \ddots & \vdots\\
-\K^1_j & -\K^2_j & -\K^3_j & \cdots & \J_{j}
\end{pmatrix},
\end{equation}
and we define the $2j\times2$ matrix
\begin{equation} \label{Lj}
\L_j=\begin{pmatrix}
b_1\1^{2,1} & c_1\1^{2,1} \\
b_2\1^{2,1} & c_2\1^{2,1} \\
\vdots & \vdots \\
b_j\1^{2,1}  & c_j\1^{2,1}
\end{pmatrix}.
\end{equation}
The transpose of a matrix will be indicated by a superscript $^T$.

\begin{proposition} \label{P1}
The homogeneous $2(j+1)$-dimensional Lotka-Volterra system \eqref{LV} with matrix
\begin{equation} \label{A1}
\A=\begin{pmatrix}
\M^j_{j+1} & \L_j \\
-\L_j^T & \J_{j+1}
\end{pmatrix}
\end{equation}
is Liouville integrable, with $j$ pairwise Poisson commuting integrals,
\begin{equation} \label{fki}
F_i=(x_{2i-1}+x_{2i})^{a_{j+1}}x_{n-1}^{-c_{i}}x_{n}^{b_{i}},\quad i=1,\ldots, j.
\end{equation}
\end{proposition}
\begin{proof}
The integrals can either by obtained from the Darboux polynomials $x_{2i-1}+x_{2i}$, $i=1,\ldots, j$, using Eq. \eqref{integral}, or by contraction to ($n-1$)-component systems and lifting their Casimirs, see \cite[Theorem 2 and Appendix B]{Poisson}. We first show their cofactors vanish. Using the notation
\begin{equation} \label{aik}
a_{i,k}=\frac{b_ic_k-b_kc_i}{a_{j+1}},
\end{equation}
we have
\begin{align}
\dot{x}_{2i-1}+\dot{x}_{2i}&=x_{2i-1}\Big(-\Big(\sum_{k=1}^{i-1}a_{k,i}(x_{2k-1}+x_{2k})\Big)
+a_ix_{2i}+\Big(\sum_{k=i+1}^{j}a_{i,k}(x_{2k-1}+x_{2k})\Big)\notag\\
&\ \ \ +b_ix_{n-1}+c_ix_{n}\Big) +x_{2i}\Big(-\Big(\sum_{k=1}^{i-1}a_{k,i}(x_{2k-1}+x_{2k})\Big)
-a_ix_{2i-1} \notag\\
&\ \ \ +\Big(\sum_{k=i+1}^{j}a_{i,k}(x_{2k-1}+x_{2k})\Big)+b_ix_{n-1}+c_ix_{n}\Big)\notag\\
&=\Big(x_{2i-1}+x_{2i}\Big)\Big(-\Big(\sum_{k=1}^{i-1}a_{k,i}(x_{2k-1}
+x_{2k})\Big)+\Big(\sum_{k=i+1}^{j}a_{i,k}(x_{2k-1}+x_{2k})\Big)\notag\\
&\ \ \ +b_ix_{n-1}+c_ix_{n}\Big)\label{cf1ehj00}\\
\dot{x}_{n-1}&=x_{n-1}\Big(-\Big(\sum_{k=1}^{j}b_{k}(x_{2k-1}+x_{2k})\Big)+a_{j+1}x_{n} \Big)\label{cf2ehj00}\\
\dot{x}_{n}&=x_{n}\Big(-\Big(\sum_{k=1}^{j}c_{k}(x_{2k-1}+x_{2k})\Big)-a_{j+1}x_{n-1} \Big),\label{cf3ehj00}
\end{align}
and hence the cofactor of $F_i$ is the linear combination of cofactors:
\begin{align*}
a_{j+1}&C[x_{2i-1}+x_{2i}]-c_iC[x_{n-1}]+b_iC[x_{n}]\\
&=-\left(\sum_{k=1}^{i-1}\left(a_{j+1} a_{k,i}-c_ib_k+b_ic_k\right)(x_{2k-1}+x_{2k})\right)-\left(-c_ib_i+b_ic_k\right)\left(x_{2i-1}+x_{2i}\right) \\
&\ \ \ +\left(\sum_{k=i+1}^{j}\left(a_{j+1} a_{i,k}+c_ib_k-b_ic_k\right)(x_{2k-1}+x_{2k})\right)+(a_{j+1}b_i+b_i(-a_{j+1}))x_{n-1}\\
&\ \ \ +(a_{j+1}c_i-c_ia_{j+1})x_{n}\\
&=0.
\end{align*}
Next, using the brackets
\begin{equation} \label{tbs}
\begin{split}
\{x_{2i-1}+x_{2i},x_{n-1}\}&=b_i(x_{2i-1}+x_{2i})x_{n-1},\\
\{x_{2i-1}+x_{2i},x_{n}\}&=c_i(x_{2i-1}+x_{2i})x_{n},\\
\{x_{2i-1}+x_{2i},x_{2k-1}+x_{2k}\}&=a_{i,k}(x_{2i-1}+x_{2i})(x_{2k-1}+x_{2k}),\\
\{x_{n-1},x_{n}\}&=a_{j+1}x_{n-1}x_{n},
\end{split}
\end{equation}
we calculate
\begin{align}
&\{F_i,F_k\}=\{(x_{2i-1}+x_{2i})^{a_{j+1}}x_{n-1}^{-c_{i}}x_{n}^{b_{i}},(x_{2k-1}+x_{2k})^{a_{j+1}}x_{n-1}^{-c_{k}}x_{n}^{b_{k}}\} \notag\\
&=a_{j+1}^2(x_{2i-1}+x_{2i})^{a_{j+1}-1}x_{n-1}^{-c_{i}}x_{n}^{b_{i}}(x_{2k-1}+x_{2k})^{a_{j+1}-1}x_{n-1}^{-c_{k}}x_{n}^{b_{k}}\{x_{2i-1}+x_{2i},x_{2k-1}+x_{2k}\}\notag\\
&-\ \ \ a_{j+1}c_{k}(x_{2i-1}+x_{2i})^{a_{j+1}-1}x_{n-1}^{-c_{i}}x_{n}^{b_{i}}(x_{2k-1}+x_{2k})^{a_{j+1}}x_{n-1}^{-c_{k}-1}x_{n}^{b_{k}}\{x_{2i-1}+x_{2i}, x_{n-1}\}\notag\\
&+\ \ \ a_{j+1}b_{k}(x_{2i-1}+x_{2i})^{a_{j+1}}x_{n-1}^{-c_{i}}x_{n}^{b_{i}}(x_{2k-1}+x_{2k})^{a_{j+1}}x_{n-1}^{-c_{k}}x_{n}^{b_{k}-1}\{x_{2i-1}+x_{2i},x_{n}\}\notag\\
&-\ \ \ c_{i}a_{j+1}(x_{2i-1}+x_{2i})^{a_{j+1}}x_{n-1}^{-c_{i}-1}x_{n}^{b_{i}}(x_{2k-1}+x_{2k})^{a_{j+1}-1}x_{n-1}^{-c_{k}}x_{n}^{b_{k}}\{x_{n-1},x_{2k-1}+x_{2k}\}\notag\\
&-\ \ \ c_{i}b_{k}(x_{2i-1}+x_{2i})^{a_{j+1}}x_{n-1}^{-c_{i}-1}x_{n}^{b_{i}}(x_{2k-1}+x_{2k})^{a_{j+1}}x_{n-1}^{-c_{k}}x_{n}^{b_{k}-1}\{x_{n-1},x_{n}\}\notag\\
&+\ \ \ b_{i}a_{j+1}(x_{2i-1}+x_{2i})^{a_{j+1}}x_{n-1}^{-c_{i}}x_{n}^{b_{i}-1}(x_{2k-1}+x_{2k})^{a_{j+1}-1}x_{n-1}^{-c_{k}}x_{n}^{b_{k}}\{x_{n},x_{2k-1}+x_{2k}\}\notag\\
&-\ \ \ b_{i}c_{k}(x_{2i-1}+x_{2i})^{a_{j+1}}x_{n-1}^{-c_{i}}x_{n}^{b_{i}-1}(x_{2k-1}+x_{2k})^{a_{j+1}}x_{n-1}^{-c_{k}-1}x_{n}^{b_{k}}\{x_{n},x_{n-1}\}\notag\\
&=F_iF_k\left(a_{j+1}^2a_{i,k}-a_{j+1}c_{k}b_i+a_{j+1}b_{k}c_i+c_{i}a_{j+1}b_k-c_{i}b_{k}a_{j+1}-b_{i}a_{j+1}c_k+b_{i}c_{k}a_{j+1}\right)\notag\\
&=0.\label{fifk}
\end{align}
\end{proof}
There are $3j+1$ parameters in the matrix \eqref{A1}. 

\subsubsection{The even-dimensional homogeneous type $[0,k,0]$ LV system} \label{S0k0}
The type $[0,k,0]$ LV system corresponds to the hypergraph in Fig. \ref{F2}.

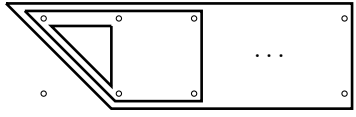
\begin{figure}[H]
\centering
\begin{tikzpicture}[scale=1]
\tikzstyle{nod}= [circle, inner sep=0pt, fill=white, minimum size=2pt, draw]		

\node[nod] (g) at (4,0) {};
\node[nod] (h) at (4,1) {};
\node[nod] (i) at (5,0) {};
\node[nod] (j) at (5,1) {};
\node[nod] (k) at (6,0) {};
\node[nod] (l) at (6,1) {};
\node[nod] (m) at (8,0) {};
\node[nod] (n) at (8,1) {};
\coordinate (js) at (4.9,0.9) {};
\coordinate (hs) at (4.1,0.9) {};
\coordinate (is) at (4.9,0.1) {};
\node at (7,1/2) {$\ldots$};
\draw[line width=1pt] (js) -- (hs) -- (is) -- (js);
\coordinate (ht) at (3.75,1.1) {};
\coordinate (it) at (4.95,-0.1) {};
\coordinate (kt) at (6.1,-0.1) {};
\coordinate (lt) at (6.1,1.1) {};
\draw[line width=1pt] (ht) -- (lt) -- (kt) -- (it) -- (ht);
\coordinate (hr) at (3.5,1.2) {};
\coordinate (ir) at (4.9,-0.2) {};
\coordinate (kr) at (8.1,-0.2) {};
\coordinate (lr) at (8.1,1.2) {};
\draw[line width=1pt] (hr) -- (lr) -- (kr) -- (ir) -- (hr);
\end{tikzpicture}
\caption{\label{F2} A hypergraph of order $2(k+1)$ and size $k$.}
\end{figure}

We define an antisymmetric $n\times n$ matrix $\N^n$ by, with $k<l$,
\begin{equation} \label{DN}
N^n_{k,l}=
\begin{cases}
a_{3j-1} & l=2j \\
a_{3j-2} & k\leq 2j, l=2j+1 \\
a_{3j} & k= 2j, l=2j+1,    
\end{cases}
\end{equation}
and $N^n_{l,k}=-N^n_{k,l}$. We denote $P_{a,b}=x_a+x_{a+1}+\cdots x_b$, following notation introduced in \cite[Section 2]{Poisson}.

\begin{proposition} \label{P2}
The homogeneous $2(k+1)$ dimensional LV system \eqref{LV} with matrix $\A=\N^{2(k+1)}$ is Liouville integrable, with $k$ pairwise commuting integrals
\begin{equation} \label{Kj}
G_j=\frac{x_{2 j}^{a_{3 j -1}}P_{1,2 j + 1}^{a_{3 j -2}-a_{3 j -1}+a_{3 j}}}{P_{1,2 j -1}^{a_{3 j}} x_{2 j +1}^{a_{3 j -2}}},\quad j=1,2,\ldots k.
\end{equation}
\end{proposition}
\begin{proof} We first prove that the $G_j$ are constants of motion.
We have
\begin{align*}
&\dot{P}_{1,r}
=\sum_{i=1}^n \left(x_1A_{1,i}+x_2A_{2,i}+\cdots+x_rA_{r,i}\right)x_{i}\\
&=\sum_{i=1}^r \left(x_1A_{1,i}+x_2A_{2,i}+\cdots+x_rA_{r,i}\right)x_{i}+
\sum_{i=r+1}^n \left(x_1A_{1,i}+x_2A_{2,i}+\cdots+x_rA_{r,i}\right)x_{i}.
\end{align*}
The first sum vanishes because $A$ is antisymmetric. From \eqref{DN} it follows that matrix $\A=\N^{2(k+1)}$ satisfies $A_{h,i}=A_{1,i}$ when $i>h+1$ or when $i=h+1$ and $h$ is odd. This implies that when $r$ is odd, the second sum factorises and we observe that $P_{1,r}$ is a DP with cofactor
\begin{equation}\label{cf1eh0k0}
C[P_{1,r}]=\sum_{i=r+1}^n A_{1,i}x_i.
\end{equation}
The cofactors of $x_{2 j}$ and $x_{2 j +1}$ are
\begin{equation}\label{cf2eh0k0}
C[x_{2j}]=A_{2j,1}P_{1,2j-1}+A_{2j,2j+1}x_{2j+1}+\sum_{i=2j+2}^{n}A_{1,i}x_i
\end{equation}
and
\begin{equation}\label{cf3eh0k0}
C[x_{2j+1}]=A_{2j+1,1}P_{1,2j-1}+A_{2j+1,2j}x_{2j}+\sum_{i=2j+2}^{n}A_{1,i}x_i.
\end{equation}
Using this, and the fact that
\begin{equation} \label{Aa3j}
A_{1,2j}=a_{3j-2},\ A_{1,2j+1}=a_{3j-1},\ A_{2j,2j+1}=a_{3j},
\end{equation}
we find
\begin{align*}
C[G_j]&=a_{3 j -1}C[x_{2 j}]+(a_{3 j -2}-a_{3 j -1}+a_{3 j})C[P_{1,2 j + 1}]-a_{3 j}C[P_{1,2 j -1}]\\
&\ \ \ -a_{3 j -2}C[x_{2 j +1}]\\
&=a_{3 j -1}\left(-a_{3j-2}P_{1,2j-1}+a_{3j}x_{2j+1}+\sum_{i=2j+2}^{n}A_{1,i}x_i
\right)\\
&\ \ \ +(a_{3 j -2}-a_{3 j -1}+a_{3 j})\left(\sum_{i=2 j + 2}^n A_{1,i}x_i\right)
-a_{3 j}\left(\sum_{i=2 j}^n A_{1,i}x_i\right)\\
&\ \ \ -a_{3 j -2}\left(-a_{3j-1}P_{1,2j-1}-a_{3j}x_{2j}+\sum_{i=2j+2}^{n}A_{1,i}x_i.
\right)\\
&=a_{3j}\left(\left(a_{3 j -1}-A_{1,2j+1}\right)x_{2j+1}+\left(a_{3 j -2}-A_{1,2j}\right)x_{2j}\right)\\
&=0.
\end{align*}
We next show that the integrals are in involution. It is easy to see that
\begin{equation} \label{pixj}
\{P_{1,i},x_j\}=A_{1,j}P_{1,i}x_j
\end{equation}
when $j>i+1$ and when $j=i+1$ with $i$ odd. This implies that, for $j>2i+1$ (we then also have $A_{2i,j}=A_{2i+1,j}=A_{1,j}$), using the Leibniz rule,
\begin{equation} \label{Kix}
\begin{split}
\{G_i,x_j\}&=\{x_{2 i}^{a_{3 i -1}}P_{1,2 i + 1}^{a_{3 i -2}-a_{3 i -1}+a_{3 i}}P_{1,2 i -1}^{-a_{3 i}} x_{2 i +1}^{-a_{3 i -2}},x_j\}\\
&=a_{3 i -1}x_{2 i}^{a_{3 i -1}-1}P_{1,2 i + 1}^{a_{3 i -2}-a_{3 i -1}+a_{3 i}}P_{1,2 i -1}^{-a_{3 i}} x_{2 i +1}^{-a_{3 i -2}}\{x_{2 i},x_j\}\\
&\ \ \ +(a_{3 i -2}-a_{3 i -1}+a_{3 i})x_{2 i}^{a_{3 i -1}}P_{1,2 i + 1}^{a_{3 i -2}-a_{3 i -1}+a_{3 i}-1}P_{1,2 i -1}^{-a_{3 i}} x_{2 i +1}^{-a_{3 i -2}}\{P_{1,2 i + 1},x_j\}\\
&\ \ \ -a_{3 i}x_{2 i}^{a_{3 i -1}}P_{1,2 i + 1}^{a_{3 i -2}-a_{3 i -1}+a_{3 i}}P_{1,2 i -1}^{-a_{3 i}-1} x_{2 i +1}^{-a_{3 i -2}}\{P_{1,2 i -1},x_j\}\\
&\ \ \ -a_{3 i -2}x_{2 i}^{a_{3 i -1}}P_{1,2 i + 1}^{a_{3 i -2}-a_{3 i -1}+a_{3 i}}P_{1,2 i -1}^{-a_{3 i}} x_{2 i +1}^{-a_{3 i -2}-1}\{x_{2 i +1},x_j\}\\
&=\left(a_{3 i -1}+(a_{3 i -2}-a_{3 i -1}+a_{3 i})-a_{3 i}-a_{3 i -2}\right) A_{1,j} x_j\\
&=0.
\end{split}
\end{equation}
From \eqref{pixj}, and linearity, it follows that
\begin{equation} \label{PiP}
\begin{split}
\{P_{1,2i-1},P_{1,2i+1}\}&=\{P_{1,2i-1},x_{2i}+x_{2i+1}\}\\
&=P_{1,2i-1}\left(A_{1,2i}x_{2i}+A_{1,2i+1}x_{2i+1}\right)\\
&=P_{1,2i-1}\left(a_{3i-2}x_{2i}+a_{3i-1}x_{2i+1}\right).
\end{split}
\end{equation}
We also have
\begin{equation} \label{xiP}
\begin{split}
\{x_{2i},P_{1,2i+1}\}&=x_{2i}\left(A_{2i,1}P_{1,2i-1}+A_{2i,2i+1}x_{2i+1}\right)\\
&=x_{2i}\left(-a_{3i-2}P_{1,2i-1}+a_{3i}x_{2i+1}\right),\\
\{x_{2i+1},P_{1,2i+1}\}&=x_{2i+1}\left(A_{2i+1,1}P_{1,2i-1}-A_{2i,2i+1}x_{2i}\right)\\
&=x_{2i+1}\left(-a_{3i-1}P_{1,2i-1}-a_{3i}x_{2i}\right).
\end{split}
\end{equation}
Using \eqref{PiP} and \eqref{xiP} we find
\begin{equation} \label{KiP}
\begin{split}
\{G_i,P_{1,2i+1}\}&=\{x_{2 i}^{a_{3 i -1}}P_{1,2 i + 1}^{a_{3 i -2}-a_{3 i -1}+a_{3 i}}P_{1,2 i -1}^{-a_{3 i}} x_{2 i +1}^{-a_{3 i -2}},P_{1,2i+1}\}\\
&=a_{3 i -1}K_i\left(-a_{3i-2}P_{1,2i-1}+a_{3i}x_{2i+1}\right)\\
&\ \ \ -a_{3 i}K_i\left(a_{3i-2}x_{2i}+a_{3i-1}x_{2i+1}\right)\\
&\ \ \ -a_{3 i -2}K_i\left(-a_{3i-1}P_{1,2i-1}-a_{3i}x_{2i}\right)\\
&=0.
\end{split}
\end{equation}
For $j>i$ the integral $G_j$ is a function of $P_{1,2i+1}$ and $x_k$ with $k>2i+1$ and therefore, by \eqref{KiP} and \eqref{Kix}, it Poisson commutes with $G_i$.
\end{proof}

\subsubsection{The even-dimensional homogeneous type $[0,k,l]$ LV system} \label{Seh0kl}
The type $[0,k,l]$ LV system corresponds to the hypergraph in Fig. \ref{F3}.

\begin{figure}[H]
\centering
\begin{tikzpicture}[scale=1]
\tikzstyle{nod}= [circle, inner sep=0pt, fill=white, minimum size=2pt, draw]		
\node[nod] (a1) at (4,0) {};
\node[nod] (b1) at (4,1) {};
\node[nod] (c1) at (5,0) {};
\node[nod] (d1) at (5,1) {};
\node[nod] (e1) at (6,0) {};
\node[nod] (h1) at (7,1) {};
\coordinate (a2) at (4.1,0.1) {};
\coordinate (b2) at (4.1,0.9) {};
\coordinate (c2) at (4.9,0.1) {};
\draw[line width=1pt] (a2) -- (b2) -- (c2) -- (a2);
\coordinate (a3) at (3.9,-0.1) {};
\coordinate (b3) at (3.9,1.1) {};
\coordinate (d3) at (5.05,1.1) {};
\coordinate (e3) at (6.25,-0.1) {};
\draw[line width=1pt] (a3) -- (b3) -- (d3) -- (e3) -- (a3);
\coordinate (a4) at (3.8,-0.2) {};
\coordinate (b4) at (3.8,1.2) {};
\coordinate (d4) at (7,1.2) {};
\coordinate (e4) at (8.4,-0.2) {};
\draw[line width=1pt] (a4) -- (b4) -- (d4) -- (e4) -- (a4);
\node[nod] (g) at (8,0) {};
\node[nod] (h) at (8,1) {};
\node[nod] (i) at (9,0) {};
\node[nod] (j) at (9,1) {};
\node[nod] (k) at (10,0) {};
\node[nod] (l) at (10,1) {};
\node[nod] (m) at (12,0) {};
\node[nod] (n) at (12,1) {};
\coordinate (js) at (8.9,0.9) {};
\coordinate (hs) at (8.1,0.9) {};
\coordinate (is) at (8.9,0.1) {};
\node at (11,1/2) {$\ldots$};
\node at (6.7,1/2) {$\ldots$};
\draw[line width=1pt] (js) -- (hs) -- (is) -- (js);
\coordinate (ht) at (7.75,1.1) {};
\coordinate (it) at (8.95,-0.1) {};
\coordinate (kt) at (10.1,-0.1) {};
\coordinate (lt) at (10.1,1.1) {};
\draw[line width=1pt] (ht) -- (lt) -- (kt) -- (it) -- (ht);
\coordinate (hr) at (7.5,1.2) {};
\coordinate (ir) at (8.9,-0.2) {};
\coordinate (kr) at (12.1,-0.2) {};
\coordinate (lr) at (12.1,1.2) {};
\draw[line width=1pt] (hr) -- (lr) -- (kr) -- (ir) -- (hr);
\end{tikzpicture}
\caption{\label{F3} A hypergraph of order $2(k+l+1)$ and size $k+l$.}
\end{figure}
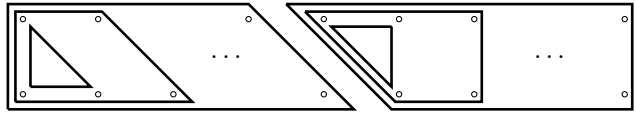

Let $\sigma_z$ be the shift operator $\sigma_z(z_i)=z_{i+1}$. We denote $\N^m_h=\sigma_a^{h-1}\N^m$, and
\begin{equation} \label{gki}
G^k_i=\sigma_x^{2k+1}\sigma_a^{3k+1}G_i.    
\end{equation}

\begin{proposition} \label{P3}
The homogeneous $2(k+l+1)$ dimensional LV system with matrix
\begin{equation} \label{Akl}
\A=\begin{pmatrix}
\N^{2k+1}_1 & a_{3k+1}\1^{2k+1,2l+1} \\
-a_{3k+1}\1^{2l+1,2k+1}  & \N^{2l+1}_{3k+2}
\end{pmatrix}
\end{equation}
is Liouville integrable. It has $3(k+l)+1$ parameters, and admits $k+l$ integrals, $G_1,\ldots,G_k$, $G^k_1,\ldots,G^k_l$, that are mutually in involution.
\end{proposition}
\begin{proof}
The proof of invariance is similar as in the proof of Prop. \ref{P2}, we adjust the cofactors \eqref{cf1eh0k0}, \eqref{cf2eh0k0} and \eqref{cf3eh0k0} as follows. For $G_j$ we have
\begin{align*}
C[P_{1,2j+1}]&=\Big(\sum_{i=2j+2}^{2k+1} A_{1,i}x_i\Big)+a_{3k+1}P_{2k+2,n}\\
C[x_{2j}]&=A_{2j,1}P_{1,2j-1}+A_{2j,2j+1}x_{2j+1}+\Big(\sum_{i=2j+2}^{2k+1}A_{1,i}x_i\Big)+a_{3k+1}P_{2k+2,n}\\
C[x_{2j+1}]&=A_{2j+1,1}P_{1,2j-1}+A_{2j+1,2j}x_{2j}+\Big(\sum_{i=2j+2}^{2k+1}A_{1,i}x_i\Big) + a_{3k+1}P_{2k+2,n}
\end{align*}
and hence 
\begin{align*}
C[G_j]&=a_{3 j -1}C[x_{2 j}]+(a_{3 j -2}-a_{3 j -1}+a_{3 j})C[P_{1,2 j + 1}]-a_{3 j}C[P_{1,2 j -1}]\\
&\ \ \ - a_{3 j -2}C[x_{2 j +1}]\\
&=a_{3 j -1}\Big(A_{2j,1}P_{1,2j-1}+A_{2j,2j+1}x_{2j+1}+\Big(\sum_{i=2j+2}^{2k+1}A_{1,i}x_i\Big)+a_{3k+1}P_{2k+2,n}\Big)\\
&\ \ \ +(a_{3 j -2}-a_{3 j -1}+a_{3 j})\Big(\Big(\sum_{i=2j+2}^{2k+1} A_{1,i}x_i\Big)+a_{3k+1}P_{2k+2,n}\Big)\\
&\ \ \ -a_{3 j}\Big(\Big(\sum_{i=2j}^{2k+1} A_{1,i}x_i\Big)+a_{3k+1}P_{2k+2,n}\Big)\\
&\ \ \ - a_{3 j -2}\Big(A_{2j+1,1}P_{1,2j-1}
+A_{2j+1,2j}x_{2j}+\Big(\sum_{i=2j+2}^{2k+1}A_{1,i}x_i\Big) + a_{3k+1}P_{2k+2,n}\Big)\\
&=\Big(a_{3 j -1}A_{2j,1}-a_{3 j -2}A_{2j+1,1}\Big)P_{1,2j-1}
-\Big(a_{3j}A_{1,2j}+a_{3 j -2}A_{2j+1,2j}\Big)x_{2j}\\
&\ \ \ +\Big(a_{3 j -1}A_{2j,2j+1}-a_{3 j} A_{1,2j+1}\Big)x_{2j+1}+\Big(a_{3 j -1}+(a_{3 j -2}-a_{3 j -1}+a_{3 j}) \\
&\ \ \ -a_{3 j} - a_{3 j -2}
\Big)\Big(\sum_{i=2j+2}^{2k+1}A_{1,i}x_i+a_{3k+1}P_{2k+2,n}\Big)\\
&=0,
\end{align*}
due to \eqref{Aa3j}. For $G^k_j$ we have
\begin{align*}
C[P_{2k+2,2(k+j+1)}]&=\sum_{i=2j+2}^{2l+1} A_{2k+2,2k+1+i}x_{2k+1+i}-a_{3k+1}P_{1,2k+1}\\
C[x_{2k+1+2j}]&=A_{2k+1+2j,2k+2}P_{2k+2,2k+2j}+A_{2k+2j+1,2k+2j+2}x_{2k+2j+2}\\
&\ \ \ +\sum_{i=2j+2}^{2l+1}A_{2k+1+1,2k+1+i}x_{2k+1+i}-a_{3k+1}P_{1,2k+1}\\
C[x_{2k+2j+2}]&=A_{2k+2j+2,2k+2}P_{2k+2,2k+2j}+A_{2k+2j+2,2k+2j+1}x_{2k+2j+1}\\
&\ \ \ +\sum_{i=2j+2}^{2l+1}A_{2k+2,2k+1+i}x_{2k+1+i} - a_{3k+1}P_{1,2k+1}.
\end{align*}
The proof that $C[G^k_j]=0$ is similar to the above, using
\[
A_{2 k +2,2 k +2 j +1}=a_{3 k +3 j -1},\ 
A_{2 k +2,2 k +2 j +2}=a_{3 k +3 j},\ 
A_{2 k +2 j +1,2 k +2 j +2}=a_{3 k +3 j +1}.
\]

The bracket $\{G_i,G_j\}$ vanishes for the same reason as before, and by relabeling of variables and parameters it also follows that $\{G^k_i,G^k_j\}=0$ (the bottom-right matrix $\N^{2l+1}_{3k+2}$ has the same form as the upper-left matrix $\N^{2k+1}_1$). It remains to establish $\{G_i,G^k_j\}=0$. One verifies that when $r$ is odd and $g>2k+1$ we have $\{P_{1,r},x_g\}=a_{3k+1}P_{1,r}x_g$. Together with the fact that the integral $G_i$ is homogeneous of weight 0, this implies that $\{G_i,x_g\}=0$. Since the same equations hold if we replace $x_g$ with $P_{2k+2,2k+2j\pm1}$, the integrals $G_i,G^k_j$ commute.
\end{proof}
\subsubsection{The even-dimensional homogeneous type $[j,k,l]$ LV system} \label{Sjkl}
The type $[j,k,l]$ LV system is the most general system we consider, it corresponds to the last hypergraph in Fig. \ref{F0}, and it includes the previous cases. We adapt the formulas \eqref{Lj}, \eqref{fki}, and \eqref{gki}:
\begin{equation} \label{Ljkl}
\L_{j,k,l}=\begin{pmatrix}
b_1\1^{2,2k+1} & c_1\1^{2,2l+1} \\
b_2\1^{2,2k+1} & c_2\1^{2,2l+1} \\
\vdots & \vdots \\
b_j\1^{2,2k+1}  & c_j\1^{2,2l+1}
\end{pmatrix},
\end{equation}
\begin{equation} \label{Fb}
\overline{F}_i=P_{2i-1,2i}^{a_{j+3k+1}}P_{2j+1,2(j+k)+1}^{-c_{i}}P_{2(j+k+1),n}^{b_{i}},
\end{equation}
with $n=2(j+k+l+1)$, and
\begin{equation} \label{gjki}
G^{j}_i=\sigma_x^{2j}\sigma_a^{j}G_i,\quad G^{j,k}_i=\sigma_x^{2(j+k)+1}\sigma_a^{j+3k+1}G_i.
\end{equation}
We denote the $2(k+l+1)\times 2(k+l+1)$ matrix \eqref{Akl} by $\A^{k,l}$ , and define $\A^{k,l}_h=\sigma_a^{h-1}\A^{k,l}$. We also use definition \eqref{Mjh} in the following proposition.

\begin{proposition} \label{P4}
The $n=2(j+k+l+1)$ dimensional $3(j+k+l)+1$ parameter family of LV systems with matrix
\begin{equation} \label{m35}
\A=\begin{pmatrix}
\M^j_{j+3k+1} & \L_{j,k,l} \\
-\L_{j,k,l}^T & \A^{k,l}_{j+1}
\end{pmatrix}
\end{equation}
is Liouville integrable. It admits the $j+k+l$ integrals $\overline{F}_1,\ldots,\overline{F}_j,G^j_1,\ldots,G^j_k,G^{j,k}_1,\ldots,G^{j,k}_l$, which are mutually in involution.
\end{proposition}
\begin{proof}
The proof that the $\overline{F}_i$ are invariant is similar to what was done in the proof of Prop. \ref{P1}, provided one redefines
\begin{equation} \label{raim}
a_{i,m}=\frac{b_ic_m-b_mc_i}{a_{j+3k+1}},
\end{equation}
and replaces the formulas \eqref{cf1eh0k0}-\eqref{cf3eh0k0} by
\begin{align*}
C[P_{2i-1,2i}]
&=
-\Big(\sum_{m=1}^{i-1}a_{m,i}P_{2m-1,2m}\Big)+\Big(\sum_{m=i+1}^{j}a_{i,m}P_{2m-1,2m}\Big)\\
&\ \ \ +b_iP_{2j+1,2(j+k)+1}+c_iP_{2(j+k+1),n}\\
C[P_{2j+1,2(j+k)+1}]&=-\Big(\sum_{m=1}^{j}b_{m}P_{2m-1,2m}\Big)+a_{j+3k+1}P_{2(j+k+1),n} \\
C[P_{2(j+k+1),n}]&=-\left(\sum_{m=1}^{j}c_{m}P_{2m-1,2m}\right)-a_{j+3k+1}P_{2j+1,2(j+k)+1}.
\end{align*}
We then find
\begin{align*}
C[\bar{F}_i]&=a_{j+3k+1}C[P_{2i-1,2i}]-c_{i}C[ P_{2j+1,2(j+k)+1}]+b_{i}C[P_{2(j+k+1),n}]\\
&=a_{j+3k+1}\Big(-\Big(\sum_{m=1}^{i-1}a_{m,i}P_{2m-1,2m}\Big)+\Big(\sum_{m=i+1}^{j}a_{i,m}P_{2m-1,2m}\Big)+b_iP_{2j+1,2(j+k)+1}\\
&\ \ \ +c_iP_{2(j+k+1),n}\Big)-c_{i}\Big(-\Big(\sum_{m=1}^{j}b_{m}P_{2m-1,2m}\Big)+a_{j+3k+1}P_{2(j+k+1),n}\Big)\\
&\ \ \ +b_{i}\Big(-\left(\sum_{m=1}^{j}c_{m}P_{2m-1,2m}\right)-a_{j+3k+1}P_{2j+1,2(j+k)+1}\Big)\\
&=\Big(-a_{j+3k+1}a_{m,i}+c_ib_m-b_ic_m\Big)\Big(\sum_{m=1}^{i-1}P_{2m-1,2m}\Big)\\
&\ \ \ +\Big(a_{j+3k+1}a_{i,m}+c_ib_m-b_ic_m\Big)\Big(\sum_{m=i+1}^{j}P_{2m-1,2m}\Big)\\
&\ \ \ +(a_{j+3k+1}b_i-b_{i}a_{j+3k+1})P_{2j+1,2(j+k)+1}
+(a_{j+3k+1}c_i-c_{i}a_{j+3k+1})P_{2(j+k+1),n}\\
&=0.
\end{align*}

The relevant cofactors for $G_i^j$ are
\begin{align*}
C[P_{2j+1,2(j+i)+1}]&=
\Big(\sum_{m=2(j+i+1)}^{2(j+k)+1)} A_{2j+1,m}x_m\Big)+a_{j+3k+1}P_{2(j+k+1),n}\\
&\ \ \ -\Big(\sum_{m=1}^{j} b_mP_{2m-1,2m}\Big)\\
C[x_r]&=-\Big(\sum_{m=1}^{j}b_{i}P_{2m-1,2m}\Big)-A_{2j+1,r}P_{2j+1,r-2}
-A_{r-1,r}x_{r-1}\\
&\ \ \ +\Big(\sum_{m=r+1}^{2(j+k)+1} A_{r,m}x_m\Big)+A_{2j+1,2(j+k+1)}P_{2(j+k+1),n}
\end{align*}
which yield
\begin{align*}
C[G_i^j]&=C\left[\frac{x_{2(j+i)}^{a_{j+3i-1}}P_{2j+1,2(j+i)+1}^{a_{j+3i-2}-a_{j+3i-1}+a_{j+3i}}}{P_{2j+1,2(j+i)-1}^{a_{j+3i}} x_{2(j+i)+1}^{a_{j+3i-2}}}\right]\\
&=a_{j+3i-1}C[x_{2(j+i)}]+(a_{j+3i-2}-a_{j+3i-1}+a_{j+3i})C[P_{2j+1,2(j+i)+1}]\\
&\ \ \ -a_{j+3i}C[P_{2j+1,2(j+i)-1}]-a_{j+3i-2}C[x_{2(j+i)+1}]\\
&=a_{j+3i-1}\Big(-\Big(\sum_{m=1}^{j}b_{i}P_{2m-1,2m}\Big)
-A_{2j+1,2(j+i)}P_{2j+1,2(j+i-1)}\\
&\ \ \ -A_{2(j+i)-1,2(j+i)}x_{2(j+i)-1}+\Big(\sum_{m=2(j+i)+1}^{2(j+k)+1} A_{2(j+i),m}x_m\Big)\\
&\ \ \ +A_{2j+1,2(j+k+1)}P_{2(j+k+1),n}\Big) +(a_{j+3i-2}-a_{j+3i-1} +a_{j+3i})\\
&\ \ \ \cdot\Big(\Big(\sum_{m=2(j+i+1)}^{2(j+k)+1} A_{2j+1,m}x_m\Big)
+a_{j+3k+1}P_{2(j+k+1),n}-\Big(\sum_{m=1}^{j} b_mP_{2m-1,2m}\Big)\Big)\\
&\ \ \ -a_{j+3i}\Big(\Big(\sum_{m=2(j+i)}^{2(j+k)+1} A_{2j+1,m}x_m\Big)+a_{j+3k+1}P_{2(j+k+1),n}-\Big(\sum_{m=1}^{j} b_mP_{2m-1,2m}\Big)\Big)\\
&\ \ \ -a_{j+3i-2}\Big(-\Big(\sum_{m=1}^{j}b_{i}P_{2m-1,2m}\Big)
-A_{2j+1,2(j+i)+1}P_{2j+1,2(j+i)-1}\\
&\ \ \ -A_{2(j+i),2(j+i)+1}x_{2(j+i)}+\Big(\sum_{m=2(j+i)+2}^{2(j+k)+1} A_{2(j+i)+1,m}x_m\Big)+A_{2j+1,2(j+k+1)}P_{2(j+k+1),n}\Big)\\
&=0,
\end{align*}
after collecting terms, and using
\begin{align*}
A_{2 j +1,2(j + i)} = a_{j +3 i -2},\
A_{2(j + i) -1,2(j + i)} &= a_{j +3 i -2},\
A_{2 j +1,2(j + i) +1} = a_{j +3 i -1},\\ 
A_{2 (j + i) ,2 (j + i) +1} = a_{j +3 i},\
A_{2 j +1,2 (j + i) +2} &=  a_{j +3 i +1}.
\end{align*}
For
\[
G_i^{j,k}=\frac{x_{2(j+k+i)+1}^{a_{j+3(k+i)}}P_{2(j+k+1),2(j+k+i+1)}^{a_{j+3(k+i)-1}-a_{j+3(k+i)}+a_{j+3(k+i)+1}}}{P_{2(j+k+1),2(j+k+i)}^{a_{j+3(k+i)+1}} x_{2(j+k+i+1)}^{a_{j+3(k+i)-1}}}
\]
we use the cofactors
\begin{align*}
C[P_{2(j+k+1),2(j+k+i)}]&=-\Big(\sum_{m=1}^{j} c_mP_{2m-1,2m}\Big)-a_{j+3k+1}P_{2j+1,2(j+k)+1}\\
&\ \ \ +
\sum_{m=2(j+k+i)+1}^{n} A_{2(j+k+1),m}x_m\\
C[x_r]&=-\Big(\sum_{m=1}^{j}c_{i}P_{2m-1,2m}\Big)-a_{j+3k+1}P_{2j+1,2(j+k)+1}\\
&\ \ \ +\sum_{m=2(j+k+1)}^{n} A_{r,m}x_m.
\end{align*}
We find
\begin{align*}
C[G_i^{j,k}]&=a_{j+3(k+i)}C[x_{2(j+k+i)+1}]\\
&\ \ \ +(a_{j+3(k+i)-1}-a_{j+3(k+i)}+a_{j+3(k+i)+1})C[P_{2(j+k+1),2(j+k+i+1)}]\\
&\ \ \ -a_{j+3(k+i)+1}C[P_{2(j+k+1),2(j+k+i)}]-a_{j+3(k+i)-1}C[x_{2(j+k+i+1)}]\\
&=a_{j+3(k+i)}\Big(-\Big(\sum_{m=1}^{j}c_{i}P_{2m-1,2m}\Big) -a_{j+3k+1}P_{2j+1,2(j+k)+1}\\
&\ \ \ +\sum_{m=2(j+k+1)}^{n} A_{2(j+k+i)+1,m}x_m\Big)+(a_{j+3(k+i)-1}-a_{j+3(k+i)}+a_{j+3(k+i)+1})\\
&\ \ \ \cdot\Big(-\Big(\sum_{m=1}^{j} c_mP_{2m-1,2m}\Big)-a_{j+3k+1}P_{2j+1,2(j+k)+1}+
\sum_{m=2(j+k+i+1)+1}^{n} A_{2(j+k+1),m}x_m\Big)\\
&\ \ \ -a_{j+3(k+i)+1}\Big(-\Big(\sum_{m=1}^{j} c_mP_{2m-1,2m}\Big)-a_{j+3k+1}P_{2j+1,2(j+k)+1}\\
&\ \ \ 
+
\sum_{m=2(j+k+i)+1}^{n} A_{2(j+k+1),m}x_m\Big)-a_{j+3(k+i)-1}\Big(-\Big(\sum_{m=1}^{j}c_{i}P_{2m-1,2m}\Big) \\
&\ \ \ 
-a_{j+3k+1}P_{2j+1,2(j+k)+1}+\sum_{m=2(j+k+1)}^{n} A_{2(j+k+i+1),m}x_m\Big)\\
&=0,
\end{align*}
by collecting terms and using, for $i=1,\ldots,l$,
\begin{align*}
A_{2(j+k+i)+1,m}&=\begin{cases} -a_{j+3(k+i)-1} &m=2(j+k+1),\ldots,2(j+k+i) \\
a_{j+3(k+i)+1} &m=2(j+k+i+1)
\end{cases}\\
A_{2(j+k+i+1),m}&=\begin{cases} -a_{j+3(k+i)} &m=2(j+k+1),\ldots,2(j+k+i) \\
-a_{j+3(k+i)+1} &m=2(j+k+i)+1
\end{cases}\\
A_{2(j+k+i),m}&=\begin{cases} a_{j+3(k+i)-1} &m=2(j+k+i)+1 \\
a_{j+3(k+i)} &m=2(j+k+i+1)
\end{cases}
\end{align*}
and, for $i=1,\ldots,l$ and $m=2(j+k+i)+3,\ldots,n$,
\[
A_{2(j+k+i)+1,m}=A_{2(j+k+i),m}=A_{2(j+k+i+1),m}.
\]
Rest us to prove the involutivity. The proof that $\{\overline{F}_m,\overline{F}_t\}=0$ is a similar calculation as \eqref{fifk}, as the commutation relations \eqref{tbs} generalise to
\begin{align*}
\{P_{2m-1,2m},P_{2j+1,n-2l-1}\}&=b_mP_{2m-1,2m}P_{2j+1,n-2l-1},\\
\{P_{2m-1,2m},P_{n-2l,n}\}&=c_m P_{2m-1,2m}P_{n-2l,n},\\
\{P_{2m-1,2m},P_{2t-1,2t}\}&=a_{m,t} P_{2m-1,2m}P_{2t-1,2t},\\
\{P_{2j+1,n-2l-1},P_{n-2l,n}\}&=a_{j+3k+1}P_{2j+1,n-2l-1}P_{n-2l,n},
\end{align*}
where $a_{m,t}$ is defined by \eqref{raim}. Since the integrals $G^j_m,G^{j,k}_t$ do not depend on the variables $x_1,\ldots,x_j$ the only relevant part of the matrix $\A$ is the bottom-right part $\A^{k,l}_{j+1}$ which is the same as in Prop. \ref{P3} but shifted. Hence the $G$-integrals Poisson commute. They also commute with the
$\overline{F}$-integrals because for all $i=1,\ldots,j$, $1\leq m\leq k$,  $1\leq t\leq l$, we have $\{x_i,G^j_m\}=\{x_i,G^{j,k}_t\}=0$ as in both cases the derivation $e\mapsto\{x_i,e\}$ acts as a scaling and the $G$-integrals are homogeneous of weight zero.
\end{proof}

\subsection{Odd-dimensional homogeneous Liouville integrable LV systems} \label{Sodd}
If the matrix $\A$ does not have full rank, Casimirs can be found by determining the null-space of $\A$ (which, as $\A$ is antisymmetric, equals the null-space of its transpose). Indeed, if $\sum_j A_{i,j}q_j=0$ for all $i$, and $\C=\prod_h x_h^{q_h}$, then
\[
\{x_i,\C\} = x_i \C \left(\sum_{j} A_{i,j} q_j \right)=0.
\]
We will obtain odd $n=2m-1$ dimensional systems from even $n+1=2m$ dimensional systems by reduction, i.e. setting $x_{2m}=0$. Integrals of the $2m$-dimensional system which do not depend on $x_{2m}$, will be integrals of the reduced system. The reduced system has one less integral, but it makes up for that by acquiring a Casimir.

\subsubsection{The odd-dimensional homogeneous $[j,0,0]$ LV system} \label{S0j00}
The matrix of the $n=2j+1$ dimensional LV system of type $[j,0,0]$ is obtained from the matrix of the type $[j,0,0]$ system by deleting the last row and last column. It has $3j+1$ parameters (the same number as the $2j+2$-dimensional type $[j,0,0]$ system). Integrals are obtained from \eqref{fki}; we define
\begin{align}
Q_i&=\left(\frac{F_i^{b_{i+1}}}{F_{i+1}^{b_i}}\right)^{1/a_{j+1}} \notag \\
&=\frac{\left(x_{2 i -1}+x_{2 i}\right)^{b_{i +1}} }{\left(x_{2 i +1}+x_{2 i +2}\right)^{b_{i}}}x_{n}^{a_{i,i+1}}, \label{Qi}  
\end{align}
where $a_{i,i+1}$ is defined in \eqref{aik}.
\begin{proposition} \label{P5}
The odd-dimensional type $[j,0,0]$ LV system is Liouville integrable, with $j-1$ integrals $Q_1,\ldots,Q_{j-1}$ mutually in involution. The Casimir is given by
\begin{equation} \label{C1}
\C=\prod_{i=1}^n x_i^{q_i}, 
\end{equation}
where $q_{2k-1}=-q_{2k}=b_k/a_k$ for $k=1,\ldots,j$, and $q_{n}=1$.
\end{proposition}
\begin{proof}
The commutativity of the $Q$-integrals follows from the commutativity of the $F$-integrals. For the Casimir we need to show that the $q$-vector is in the null-space of $\A$. We have
\begin{align*}
\sum_{l=1}^{n} A_{2k-1,l} q_l&=
\Big(\sum_{l=1}^{2k-2} A_{2k-1,l} q_l\Big) +
A_{2k-1,2k-1}q_{2k-1}+A_{2k-1,2k}q_{2k}\\
&\ \ \ +
\Big(\sum_{l=2k+1}^{n-1} A_{2k-1,l} q_l\Big)
+ A_{2k-1,n} q_n\\
&=
0 + 0 + a_k (-\frac{b_k}{a_k}) +
0
+ b_k\\
&=0,
\end{align*}
as the terms in the summations cancel each other pairwise. The proof that $\sum_{l=1}^{n} A_{2k,l} q_l=0$ is similar.
\end{proof}

\subsubsection{The odd-dimensional homogeneous $[0,k,0]$ LV system}
\begin{proposition} \label{P6}
A Casimir for the $n=2k+1$ dimensional LV system of type $[0,k,0]$, with matrix $\A=\N^{2k+1}$ defined by Eq. \eqref{DN}, is given by
\begin{equation}\label{cq2}
\C=\prod_{i=1}^n x_i^{q_i},\quad q_i=\begin{cases}
a_3 & i=1 \\
-a_{3j-1}\prod_{l=1}^{j-1} \frac{a_{3l-2}-a_{3l-1}+a_{3l}}{a_{3l+3}} & i=2j \\
a_{3j-2}\prod_{l=1}^{j-1} \frac{a_{3l-2}-a_{3l-1}+a_{3l}}{a_{3l+3}} & i=2j+1.
\end{cases}
\end{equation}
The system is Liouville integrable, with $k-1$ pairwise commuting integrals given by $G_j$, $j=1,2,\ldots k-1$, as given by Eq. \eqref{Kj}.
\end{proposition}
\begin{proof} The commutativity of the $G$-integrals was proven in Prop. \ref{P2}. We show that the $q$-vector in \eqref{cq2} is in the null-space of $\A$. Let $i$ be even. We claim that
\begin{equation} \label{cl1}
\sum_{j=1}^{i+1} A_{i,j}q_j = 0,\quad \sum_{j=i+2}^{n} A_{i,j}q_j = 0,
\end{equation}
from which the result follows. The second equality of \eqref{cl1} holds as the terms vanish pairwise. The first equality of \eqref{cl1} boils down to
\[
  \sum_{j=1}^{i-1} q_j = \frac{a_{3i/2}}{a_{3i/2-2}} q_{i+1}
\]
which can be proven by induction. Similarly, when $i$ is odd we have $\sum_{j=1}^{i-1} A_{i,j}q_j = 0$ and $\sum_{j=i+1}^{n} A_{i,j}q_j = 0$.
\end{proof}

\subsubsection{The odd-dimensional homogeneous $[0,k,l]$ LV system} \label{Soh0kl}
The matrix $\A$ of the $n=2(k+l)+1$ dimensional LV system of type $[0,k,l]$ is obtained from the matrix of the type $[0,k,l]$ system by deleting the last row and last column. It has $3(k+l)-1$ parameters.
\begin{proposition} \label{P7}
The odd-dimensional type $[0,k,l]$ LV system is Liouville integrable, with $k+l-1$ commuting integrals $G_1,\ldots,G_k,G^k_1,\ldots,G^k_{l-1}$, given in Eq. \eqref{Kj} and Eq. \eqref{gki}. The Casimir is given by $\C=\prod_{i=1}^n x_i^{q_i}$, with 
\begin{equation}\label{cq3}
q_i=\begin{cases}
a_{3 k +3 l -1} \left(\prod_{j =1}^k\frac{a_{3 j}}{a_{3 j -2}-a_{3 j -1}+a_{3 j}}\right), & i=1 \\
-\frac{a_{3 k +3 l -1} a_{3 m -1}}{a_{3 m}} \left(\prod_{j =m}^k \frac{a_{3 j}}{a_{3 j -2}-a_{3 j -1}+a_{3 j}}\right), & i=2m,\ m=1,\ldots,k \\
\frac{a_{3 k +3 l -1} a_{3 m -2}}{a_{3 m}} \left(\prod_{j =m}^k \frac{a_{3 j}}{a_{3 j -2}-a_{3 j -1}+a_{3 j}}\right), & i=2m+1,\ m=1,\ldots,k \\
-a_{3 k +1} \left(\prod_{j =k +1}^{k +l -1} \frac{a_{3 j +1}}{a_{3 j -1}-a_{3 j}+a_{3 j +1}}\right), & i=2k+2 \\
\frac{a_{3 k +1} a_{3 m}}{a_{3 m +1}} \left(\prod_{j =m}^{k +l -1} \frac{a_{3 j +1}}{a_{3 j -1}-a_{3 j}+a_{3 j +1}}\right), & i=2m+1,\ m=k+1,\ldots,k+l-1 \\
-\frac{a_{3 k +1} a_{3 m -1}}{a_{3 m +1}} \left(\prod_{j =m}^{k +l -1} \frac{a_{3 j +1}}{a_{3 j -1}-a_{3 j}+a_{3 j +1}}\right), & i=2m+2,\ m=k+1,\ldots,k+l-1 \\
a_{3 k +1}, & i=2(k+l)+1.
\end{cases}
\end{equation}
\end{proposition}
\begin{proof}
The commuting integrals are the same as those in Prop. \ref{P3}. We have
\begin{equation} \label{qqq}
\begin{cases}
    \sum_{j=1}^{i-1} A_{i,j}q_j
    = \sum_{j=i+1}^{2k+1} A_{i,j}q_j
    = 0 &1 \equiv i \leq 2k+1, \\
    \sum_{j=1}^{i+1} A_{i,j}q_j
    = \sum_{j=i+2}^{2k+1} A_{i,j}q_j
    = 0 &0 \equiv i \leq 2k+1, \\
    \sum_{j=2k+2}^{n} A_{i,j}q_j = 0 & i \leq 2k+1, \\
    \sum_{j=1}^{2k+1} A_{i,j}q_j = - A_{i,n}q_n & i > 2k+1, \\
    \sum_{j=2k+2}^{i+1} A_{i,j}q_j
    = \sum_{j=i+2}^{n-1} A_{i,j}q_j
    = 0 &1 \equiv i > 2k+1, \\
    \sum_{j=2k+2}^{i-1} A_{i,j}q_j
    = \sum_{j=i+1}^{n-1} A_{i,j}q_j
    = 0 &0 \equiv i > 2k+1,
\end{cases}
\end{equation}
from which it follows that the $q$-vector \eqref{cq3} is in the null-space of $\A$.
\end{proof}
\subsubsection{The odd-dimensional homogeneous $[j,k,l]$ LV system} \label{224}
The matrix of the $n=2(j+k+l)+1$ dimensional LV system of type $[j,k,l]$ is obtained from the matrix of the even-dimensional type $[j,k,l]$ system, by deleting the last row and last column. It has $3(j+k+l)-1$ parameters. We define
\begin{align}
\overline{Q}_i&=\left(\frac{\overline{F}_i^{b_{i+1}}}{\overline{F}_{i+1}^{b_i}}\right)^{1/a_{j+3k+1}} \notag \\
&=\frac{P_{2 i -1,2 i}^{b_{i +1}} }{P_{2 i +1,2 i +2}^{b_{i}}}P_{2j+1,2(j+k)+1}^{a_{i,i+1}}, \label{oQi}
\end{align}
where $a_{i,i+1}$ is defined in \eqref{raim}.
\begin{proposition} \label{P8}
The odd-dimensional type $[j,k,l]$ LV system is Liouville integrable. The Casimir is
$\C=\prod_{i=1}^n x_i^{q_i}$, with 
\begin{equation}\label{cq4}
q_i=\begin{cases}
\frac{b_m}{a_m} & i=2m-1,\ m=1,\ldots,j  \\
-\frac{b_m}{a_m} & i=2m,\ m=1,\ldots,j  \\
\prod_{h=1}^k \frac{a_{j +3 h}}{a_{j +3 h -2}-a_{j +3 h -1}+a_{j +3 h}} & i=2j+1 \\
-\frac{a_{j +3 m -1}}{a_{j +3 m}}
\prod_{h=m}^k \frac{a_{j +3 h}}{a_{j +3 h -2}-a_{j +3 h -1}+a_{j +3 h}}
& i=2j+2m,\ m=1,\ldots,k\\
\frac{a_{j +3 m -2}}{a_{j +3 m}}
\prod_{h=m}^k \frac{a_{j +3 h}}{a_{j +3 h -2}-a_{j +3 h -1}+a_{j +3 h}}
& i=2j+2m+1,\ m=1,\ldots,k\\
-\frac{a_{j +3 k +1}}{a_{j +3 k +3 l -1}}
\prod_{h=k+1}^{k+l-1} \frac{a_{j +3 h +1}}{a_{j +3 h -1}-a_{j +3 h}+a_{j +3 h +1}}
& i=2(j+k)+2\\
\frac{a_{j +3 k +1} a_{j +3 k +3m}}{a_{j +3 k +3 l -1} a_{j +3 k +3m+1}}
\prod_{h=k+m}^{k+l-1}
\frac{a_{j +3 h +1}}{a_{j +3 h -1}-a_{j +3 h}+a_{j +3 h +1}}
& i=2(j+k+m)+1,\ m=1,\ldots,l-1\\
-\frac{a_{j +3 k +1} a_{j+3k+3m-1}}{a_{j +3 k +3 l -1} a_{j +3 k +3m +1}}
\prod_{h=k+m}^{k+l-1}
\frac{a_{j +3 h +1}}{a_{j +3 h -1}-a_{j +3 h}+a_{j +3 h +1}}
& i=2(j+k+m)+2,\ m=1,\ldots,l-1\\
\frac{a_{j +3 k +1}}{a_{j +3 k +3 l -1}}
& i=2(j+k+l)+1,
\end{cases}
\end{equation}
and $j+k+l-1$ commuting integrals  are $\overline{Q}_1,\ldots,\overline{Q}_{j-1},G^j_1,\ldots,G^j_k,G^{j,k}_1,\ldots,G^{j,k}_{l-1},K$, defined by Eqs. \eqref{oQi}, \eqref{gki}, \eqref{gjki} and
\begin{equation} \label{KK}
K=\prod_{i=1}^j \left(\frac{x_{2i-1}}{x_{2i}}\right)^{b_i/a_i}
\left(x_1+x_2\right)^{a_{j+3k+1}/c_1}
\left(x_{2(j+k+1)}+\cdots+x_n\right)^{b_1/c_1}\prod_{i=2(j+k+1)}^n x_i^{q_i}.
\end{equation}
\end{proposition}
\begin{proof}  The commuting integrals $\overline{Q}_1,\ldots,\overline{Q}_{j-1},G^j_1,\ldots,G^j_k,G^{j,k}_1,\ldots,G^{j,k}_{l-1}$ are inherited from Prop. \ref{P4}. The $q$-vector in \eqref{cq4} is in the null-space of $\A$, due to
\begin{equation} \label{qqqqq}
\begin{cases}
   - \sum_{h=1}^{2j} A_{i,h}q_h = \sum_{h=2j+1}^{2j+2k+1} A_{i,h}q_h = b_{\lfloor(i+1)/2\rfloor} & i \leq 2j \\
\sum_{h=2(j+k+1)}^n A_{i,h}q_h = 0 & i \leq 2j \\
\sum_{h=1}^{2j} A_{i,h}q_h = 0 & i > 2j \\
\sum_{h=2j+1}^{2(j+k)+1} A_{i,h}q_h = \sum_{h=2(j+k+1)}^n A_{i,h}q_h = 0 & 2j < i \leq 2(j+k)+1 \\
- \sum_{h=2j+1}^{2(j+k)+1} A_{i,h}q_h = \sum_{h=2(j+k+1)}^n A_{i,h}q_h = a_{j+3k+1} & i > 2(j+k)+1.
\end{cases}
\end{equation}
The function $K$ is an integral, as we have
\begin{align*}
C[x_{2i-1}/x_{2i}]&=a_i(x_{2i-1}+x_{2i})\\
C[x_1+x_2]&=\sum_{h=3}^nA_{1,h}x_h\\
C[x_{2(j+k+1)}+\cdots+x_n]&=-a_{j+3k+1}(x_{2j+1}+\cdots+x_{2(j+k)+1})-\sum_{i=1}^j c_i(x_{2i-1}+x_{2i})\\
C[\prod_{i=2(j+k+1)}^n x_i^{q_i}]&=-a_{j+3k+1}(x_{2(j+k+1)}+\cdots+x_{n}),
\end{align*}
and hence
\begin{align*}
C[K]&=\sum_{i=1}^j b_i(x_{2i-1}+x_{2i}) + \frac{a_{j+3k+1}}{c_1} \Big( \sum_{i=2}^j \frac{b_1c_i-b_ic_1}{a_{j+3k+1}}(x_{2i-1}+x_{2i})+b_1(x_{2j+1}+\cdots\\
&\ \ \ 
+x_{2(j+k)+1})+c_1(x_{2(j+k+1)}+\cdots+x_{n})\Big) - \frac{b_1}{c_1}\Big(a_{j+3k+1}(x_{2j+1}+\cdots\\
&\ \ \ +x_{2(j+k)+1})+\sum_{i=1}^j c_i(x_{2i-1}+x_{2i})\Big) -a_{j+3k+1} (x_{2(j+k+1)}+\cdots+x_{n})\\
&=0.
\end{align*}
The integral $K$ commutes with the integrals $\overline{Q}_1,\ldots,\overline{Q}_{j-1},G^j_1,\ldots,G^j_k$, because
\[
\{x_1+x_2,K\}=0,\quad \{x_i,K\}=0,\quad i=3,\ldots,2(j+k)+1,
\]
and it commutes with the integrals $G^{j,k}_1,\ldots,G^{j,k}_{l-1}$ due to
\[
\{x_i,K\}=\frac{b_1}{c_1}\frac{K}{x_{2(j+k+1)}+\cdots+x_n}x_i\sum_{m=2(j+k+1)}^n
A_{i,m}x_m,\quad i=2(j+k)+2,\ldots,n.
\]
\end{proof}
We note that the type $[j,k,0]$ system arises as a special case of the above.
\section{Nonhomogeneous Liouville integrable LV systems} \label{Snonh}
The family of Lotka-Volterra equations \eqref{LV} is Hamiltonian, or Poisson, if $\A$ is skew and $\b\in\text{Im}(\A)$. In terms of the Poisson bracket \eqref{PBA}, the system (\ref{LV}) can be written as $\dot{\x}=\{\x,H\}$, where the Hamiltonian is given by \cite{FerOli,MQ}
\begin{equation}\label{hamLV}
H = \sum_{i=1}^{n} (x_i + k_i \ln x_i), \qquad \text{ with } \A\k=\b.
\end{equation}

A function $P(\x)$ is homogeneous of weight $d$ if $P(\lambda\x)=\lambda^dP(\x)$.
The following theorem is an improved version of \cite[Theorem 3.1]{KMQ}. Let $I=\{i_1,\ldots,i_k\}\subset \mathbb{N}_n$, and denote $\x_I=(x_{i_1},\ldots,x_{i_k})$. 
\begin{theorem} \label{GODEs}
Let $P(\x_I)$ be a homogeneous second integral of weight $d$ with cofactor $C(\x)$ for the system of ODEs $\dot{\x}=\f(\x)$. Then $P(\x_I)$ is a second integral for the system $\dot{\x}=\f(\x)+r(\x,t)\x$, with cofactor $C(\x)+dr(\x,t)$, where $r$ is a scalar function of $\x,t$.
\end{theorem}

It follows that if a homogeneous LV system admits an integral $F$ which is homogeneous of weight 0, then the inhomogeneous LV system obtained by adding $\b=r\x$ to the right-hand-side admits the same integral. Note that, in Theorem \ref{GODEs}, the only relevant equations are the ones for $\x_I$, the equations for the complementary variables $\x_{I^c}$ can be altered differently.

\subsection{Even-dimensional inhomogeneous Liouville integrable LV systems}
\subsubsection{The even-dimensional inhomogeneous $[j,0,0]$ LV system} \label{SNj00}
The $n=2(j+1)$ dimensional inhomogeneous LV system of type $[j,0,0]$, of the form \eqref{LV} with constant $r_i=r$, has matrix
\begin{equation} \label{mnhj00ei}
\A=\begin{pmatrix}
J_1 & K_{1,2} & K_{1,3} & \cdots & K_{1,j} & L_1 \\
-K_{1,2} & J_2 & K_{2,3} & \cdots & K_{2,j} & L_2 \\
-K_{1,3} & - K_{2,3} & J_3 & \cdots & K_{3,j} & L_3 \\
\vdots & \vdots & \vdots & \ddots & \vdots & \vdots \\
-K_{1,j} & -K_{2,j} & -K_{3,j} & \cdots & J_{j} & L_{j} \\
-L_1^T & -L_2^T & -L_3^T & \cdots & -L_{j}^T & J_{j+1}
\end{pmatrix}
\end{equation}
where
\begin{equation} \label{JKLeij00}
J_i=\begin{pmatrix}
0 & a_i \\
-a_i & 0
\end{pmatrix},\
K_{i,k}=(c_i-c_k)\1^{2,2},\
L_i=\begin{pmatrix}
c_i - a_{j+1} & c_i \\
c_i - a_{j+1} & c_i
\end{pmatrix}.
\end{equation}
\begin{proposition}
The $n$-dimensional inhomogeneous LV system of type $[j,0,0]$ is a Hamiltonian system, with bracket \eqref{PBA} and Hamiltonian
\[
H=\sum_{i=1}^n x_i + \frac{r}{a_{j+1}} \ln\left(\frac{x_n}{x_{n-1}}\right),
\]
and pairwise commuting integrals
\[
F_i=(x_{2i-1} + x_{2i})^{a_{j+1}}x_{n-1}^{-c_i}x_n^{(c_i - a_{j+1})},\quad i=1,2,\ldots, j.
\]
\end{proposition}
\begin{proof}
With $\k=r(0,\ldots,0,-1,1)/a_{j+1}$ we have $(\A\k)_i=r$ for all $i$. The matrix \eqref{mnhj00ei} is obtained from \eqref{A1} by substitution of $b_i=c_i-a_{j+1}$. Similarly, the commuting integrals are obtained from \eqref{fki}. The functions $F_i$ have become homogeneous of weight 0 and hence, due to Theorem \ref{GODEs}, they remain integrals after adding $\b=r\x$. 
\end{proof}
\subsubsection{The even-dimensional inhomogeneous type $[j,k,l]$ LV system} \label{SNjkl}
The integrals \eqref{gjki} are homogeneous of weight 0. The only adjustment to be made to the $[j,k,l]$ LV system is the substitution $b_i=c_i-a_{j+3k+1}$, $i=1,\ldots,j$, so that the integrals \eqref{Fb}, become
\begin{equation} \label{Fh}
\overline{F}_i=P_{2i-1,2i}^{a_{j+3k+1}}P_{2j+1,2(j+k)+1}^{-c_{i}}P_{2(j+k+1),n}^{c_i-a_{j+3k+1}}.
\end{equation}
\begin{proposition}
The inhomogeneous $n=2(j+k+l+1)$ dimensional LV system of type $[j,k,l]$ has matrix \eqref{m35}, in which $b_i=c_i-a_{j+3k+1}$, $i=1,\ldots,j$ and a Hamiltonian of the form \eqref{hamLV}, with
\[
k_i=\begin{cases}
0  &i\leq 2j,\\
-\frac{r}{a_{j+3k+1}} \prod_{h=1}^k \frac{a_{j+3h}}{a_{j+3h-2}-a_{j+3h-1}+a_{j+3h}} &i=2j+1,\\
\frac{ra_{j+3m-1}}{a_{j+3k+1}a_{j+3m}} \prod_{h=m}^k \frac{a_{j+3h}}{a_{j+3h-2}-a_{j+3h-1}+a_{j+3h}} &i=2(j+m), m=1,\ldots,k,\\
-\frac{ra_{j+3m-2}}{a_{j+3k+1}a_{j+3m}} \prod_{h=m}^k \frac{a_{j+3h}}{a_{j+3h-2}-a_{j+3h-1}+a_{j+3h}} &i=2(j+m)+1, m=1,\ldots,k,\\
\frac{r}{a_{j+3k+1}} \prod_{h=k+1}^{k+l} \frac{a_{j+3h}}{a_{j+3h-2}-a_{j+3h-1}+a_{j+3h}} &i=2(j+k+1),\\
-\frac{ra_{j+3(k+m)}}{a_{j+3k+1}a_{j+3(k+m)+1}} \prod_{h=k+m}^{k+l} \frac{a_{j+3h}}{a_{j+3h-2}-a_{j+3h-1}+a_{j+3h}} &i=2(j+k+m)+1, m=1,\ldots,l,\\
\frac{ra_{j+3(k+m)-1}}{a_{j+3k+1}a_{j+3(k+m)+1}} \prod_{h=k+m}^{k+l} \frac{a_{j+3h}}{a_{j+3h-2}-a_{j+3h-1}+a_{j+3h}} &i=2(j+k+m+1), m=1,\ldots,l.
\end{cases}
\]
It admits $j+k+l$ integrals $\overline{F}_1,\ldots,\overline{F}_j,G^j_1,\ldots,G^j_k,G^{j,k}_1,\ldots,G^{j,k}_l$, which are mutually in involution.
\end{proposition}
\begin{proof}
The proof is similar to the previous one. The relation $\A\k=\b$ follows from
\begin{align*}
  \sum_{h=2j+1}^{2(j+k)+1} A_{i,h}k_h &= \begin{cases}
r-rc_{\lfloor (i+1)/2 \rfloor}/a_{j+3k+1}  &i\leq 2j,\\
0  &2j<i\leq 2(j+k)+1,\\
r  & 2(j+k)+1<i \leq n,  
  \end{cases} \\
  \sum_{h=2(j+k+1)}^n A_{i,h}k_h &= \begin{cases}
rc_{\lfloor (i+1)/2 \rfloor}/a_{j+3k+1}  &i\leq 2j,\\
r  &2j<i\leq 2(j+k)+1,\\
0  & 2(j+k)+1<i \leq n.  
  \end{cases}
\end{align*}
\end{proof}
We note that the LV systems of type $[0,k,l]$ are a special case of the $[j,k,l]$ system. 

\subsection{Odd-dimensional inhomogeneous LV systems}
For odd-dimensional $[j,k,l]$ LV systems, the vector $\b=r\1$ is not in the image of $\A$. We will consider vectors that are in the image of $\A$ of the form $\b=(r,\ldots,r,s)$ where $s$ may depend on the parameters of the system. This means that, according to Theorem \ref{GODEs}, homogeneous integrals of weight 0 that do not depend on $x_n$ survive inhomogenisation.  

\subsubsection{The odd-dimensional inhomogeneous $[j,0,0]$ LV system}
The $n=2j+1$ dimensional inhomogeneous LV system of type $[j,0,0]$ has matrix
\begin{equation} \label{mnhj00oi}
\A=\begin{pmatrix}
J_1 & K_{1,2} & K_{1,3} & \cdots & K_{1,j} & L_1 \\
-K_{1,2} & J_2 & K_{2,3} & \cdots & K_{2,j} & L_2 \\
-K_{1,3} & - K_{2,3} & J_3 & \cdots & K_{3,j} & L_3 \\
\vdots & \vdots & \vdots & \ddots & \vdots & \vdots \\
-K_{1,j} & -K_{2,j} & -K_{3,j} & \cdots & J_{j} & L_{j} \\
-L_1^T & -L_2^T & -L_3^T & \cdots & -L_{j}^T & J_{j+1}
\end{pmatrix}
\end{equation}
where
\begin{equation} \label{JKLo}
J_i=\begin{pmatrix}
0 & a_i \\
-a_i & 0
\end{pmatrix},\
K_{i,k}=(c_i-c_k)\1^{2,2},\
L_i=a_{j+1}\1^{2,1}.
\end{equation}
It is obtained by deleting the last row and column of \eqref{A1}, and the substitution $b_i=a_{j+1}$, $i=1,\ldots,j$. We obtain 1 Casimir, $\C=\prod_{i=1}^n x_i^{q_i}$, with 
\[
q_i=\begin{cases}
1/a_m & i=2m-1,\quad m=1,\ldots,j+1, \\
-1/a_m & i=2m,\quad m=1,\ldots,j.
\end{cases}
\]
The vector $\b=r\1$ is not in the image of $\A$, but we can take
\[
r_i=\begin{cases} r & i < n \\
0 & i=n. \end{cases}
\]
The Hamiltonian is $H=\sum_{i=1}^{2j+1} x_i + k_i \ln(x_i)$, with
\[
k_i=\begin{cases}
-rq_i& i<n\\
0 & i=n.
\end{cases}
\]
and the functions
\begin{align}
R_i&=\C\left(\frac{F_i}{F_{i+1}}\right)^{1/(a_{j+1}(c_i-c_{i+1}))} \notag \\
&=\left(\frac{x_{2 i -1}+x_{2 i}}{x_{2i+1}+x_{2i+2}}
\right)^{1/(c_i-c_{i+1})} \prod_{i=1}^j \left(
\frac{x_{2 i -1}}{x_{2 i}}\right)^{1/a_i} \label{Ri}
\end{align}
are pairwise commuting integrals that do not depend on the variable $x_n$. This yields the following result.
\begin{proposition} \label{noj}
The inhomogeneous $n=2j+1$ dimensional $[j,0,0]$ LV system is Liouville integrable.
\end{proposition}
We note that there is one additional DP, $P_{1,n-1}$, which yields an extra integral,
\[
\frac{P_{1,n-1}}{P_{1,2}}\prod_{i=2}^j \left(
\frac{x_{2 i -1}}{x_{2 i}}\right)^{(c_i-c_1)/a_i}.
\] 

\subsubsection{The odd-dimensional inhomogeneous $[0,k,0]$ LV system}
To the homogeneous $n=2k+1$ dimensional $[0,k,0]$ LV system we can add a inhomogeneous term with
\begin{equation} \label{bbb}
r_i=\begin{cases} r & i<n \\
r\frac{a_{3k-1} - a_{3k}}{a_{3k-2}}
 & i=n. \end{cases}
\end{equation}
The Hamiltonian is
\[
H=\sum_{i=1}^{2j+1} x_i + r\frac{a_{3k} - a_{3k-1}}{a_{3k-2}a_{3k}} \ln(x_{n-1})
+\frac{r}{a_{3k}}  \ln(x_{n}).
\]
Note that the highest index $h$ of $x_h$ in the function $G_j$ as defined by \eqref{Kj} is $2j+1<n$ for $j\leq k-1$.
\begin{proposition} \label{nok}
The inhomogeneous $n=2k+1$ dimensional $[0,k,0]$ LV system is Liouville integrable, with $k-1$ pairwise commuting integrals $G_j$, $j=1,2,\ldots k-1$, as given by Eq. \eqref{Kj}, and Casimir \eqref{cq2}.
\end{proposition}

\subsubsection{The odd-dimensional inhomogeneous $[0,k,l]$ LV system}
To the $n=2(k+l)+1$ dimensional $[0,k,l]$ LV system ($l\neq0$) we can add a inhomogeneous term with
\begin{equation} \label{nonht}
r_i=\begin{cases} r & i<n \\
r\frac{a_{3k+1} - a_{3(k+l)-1}}{a_{3k+1}}
 & i=n. \end{cases}
\end{equation}
\begin{proposition} \label{nokl}
The inhomogeneous $n=2(k+l)+1$ dimensional $[0,k,l]$ LV system, with \eqref{nonht}, is Liouville integrable. The Hamiltonian has the form \eqref{hamLV}, with
\begin{equation} \label{kkk}
k_i=\begin{cases}
0 & i\leq 2k+1 \\
r\frac{a_{3k+1} - a_{3(k+l)-1}}{a_{3k+1}^2a_{3(k+l)-1}}
q_i & 2k+1 < i < n\\
\frac{r}{a_{3k+1}a_{3(k+l)-1}} & i=n,
\end{cases}
\end{equation}
there are $k+l-1$ commuting integrals $G_1,\ldots,G_k,G^k_1,\ldots,G^k_{l-1}$, given in Eq. \eqref{Kj} and Eq. \eqref{gki}, and one Casimir \eqref{cq3}.
\end{proposition}
\begin{proof}
The proof that $\A\k=\b$, with $\k$ given by \eqref{kkk} and $\b$ given by \eqref{bbb}, relies on
\begin{align*}
\sum_{h=1}^{2k+1} A_{i,h}q_h &=
\begin{cases}
0 & i\leq 2k+1 \\
-a_{3k+1}a_{3(k+l)-1} & i>2k+1,
\end{cases}\\
\sum_{h=2k+1}^{n-1} A_{i,h}q_h &=
\begin{cases}
-a_{3k+1}^2 & i\leq 2k+1 \\
0 & 2k+1<i<n\\
a_{3k+1}a_{3(k+l)-1} & i=n,
\end{cases}\\
A_{i,n}q_n &=
\begin{cases}
a_{3k+1}^2 & i\leq 2k+1 \\
a_{3k+1}a_{3(k+l)-1} & 2k+1<i<n\\
0 & i=n,
\end{cases}
\end{align*}
which are consequenses of \eqref{qqq}.
\end{proof}

\subsubsection{The odd-dimensional inhomogeneous $[j,k,l]$ LV system}
The matrix of the inhomogeneous $n=2(j+k+l)+1$ dimensional $[j,k,l]$ LV system is obtained from the inhomogeneous even-dimensional $[j,k,l]$ system described in section \ref{SNjkl}, by deleting the last row and column. The matrix is obtained from the homogeneous case by substitution of $b_i=c_i-a_{j+3k+1}$. The inhomogeneous term is given by
\begin{equation} \label{nonhtjkl}
r_i=\begin{cases} r & i<n \\
r\frac{a_{j+3k+1} - a_{j+3(k+l)-1}}{a_{j+3k+1}}
 & i=n. \end{cases}
\end{equation}
\begin{proposition} \label{gnojkl}
Let $jkl\neq0$. The inhomogeneous $n=2(j+k+l)+1$ dimensional $[j,k,l]$ LV system with inhomogeneous term given by \eqref{nonhtjkl}, has a Hamiltonian of the form \eqref{hamLV}, with
\begin{equation} \label{kkkkk}
k_i=\begin{cases}
-r\frac{a_{j+3k+1}-c_m}{a_{j+3k+1}a_m}
 & i=2m-1, m=1,\ldots,j\\
r\frac{a_{j+3k+1}-c_m}{a_{j+3k+1}a_m}
 & i=2m, m=1,\ldots,j\\
0 & i=2j+1,\ldots,2(j+k)+1\\
r\frac{a_{j+3k+1} - a_{j+3(k+l)-1}}{a_{j+3k+1}^2}
q_i & i=2(j+k)+2,\ldots, n-1\\
\frac{r}{a_{j+3(k+l)-1}} & i=n.
\end{cases}
\end{equation}
There are $j+k+l-2$ commuting integrals $\overline{Q}_1,\ldots,\overline{Q}_{j-1},G_1,\ldots,G_k,G^k_1,\ldots,G^k_{l-1}$, where
\begin{align}
\overline{Q}_i&=\left(\frac{\overline{F}_i^{c_{i+1}-a_{j+3k+1}}}{\overline{F}_{i+1}^{c_i-a_{j+3k+1}}}\right)^{1/a_{j+3k+1}} \notag \\
&=\frac{P_{2 i -1,2 i}^{c_{i+1}-a_{j+3k+1}} }{P_{2 i +1,2 i +2}^{c_i-a_{j+3k+1}}} P_{2j+1,2(j+k)+1}^{c_i-c_{i+1}}, \label{noQi}
\end{align}
and the functions $G$ are given by Eqs. \eqref{Kj}, \eqref{gki}, and there is one Casimir \eqref{cq3}.
\end{proposition}
\begin{proof}
The proof that $\A\k=\b$, with $\k$ given by \eqref{kkkkk} and $\b$ given by \eqref{nonhtjkl}, relies on
\begin{align*}
\sum_{h=1}^{2j} A_{i,h}q_h &=
\begin{cases}
a_{j+3k+1}-c_{\lfloor(i+1)/2\rfloor} & i\leq 2j \\
0 & i>2j,
\end{cases}\\
\sum_{h=2j+1}^{2(2+k)+1} A_{i,h}q_h &=
\begin{cases}
c_{\lfloor(i+1)/2\rfloor}-a_{j+3k+1} & i\leq 2j \\
0 & 2j<i\leq 2(j+k)+1 \\
-a_{j+3k+1} & 2(j+k)+1<i\leq n,
\end{cases}\\
\sum_{h=2(j+k+1)}^{n-1} A_{i,h}q_h &=
\begin{cases}
-a_{3k+1}c_{\lfloor(i+1)/2\rfloor}/a_{3(k+l)-1}& i\leq 2j \\
-a_{3k+1}^2/a_{3(k+l)-1} & 2j<i\leq 2(j+k)+1 \\
0 & 2(j+k)+1<i< n\\
a_{3k+1}& i=n,
\end{cases}\\
A_{i,n}q_n &=
\begin{cases}
a_{3k+1}c_{\lfloor(i+1)/2\rfloor}/a_{3(k+l)-1} & i\leq 2j \\
a_{3k+1}^2/a_{3(k+l)-1} & 2j<i\leq 2(j+k)+1 \\
a_{3k+1}& 2(j+k)+1<i< n\\
0 & i=n,
\end{cases}
\end{align*}
which are consequences of \eqref{qqqqq}.
\end{proof}

We note that the function $K$, given by \eqref{KK}, is not an integral for the inhomogeneous odd $[j,k,l]$ system. The number of integrals for the general odd $[j,k,l]$ system, with inhomogeneous term \eqref{nonhtjkl}, given in Prop. \ref{gnojkl} is not sufficient for Liouville integrability.
For the odd-dimensional $[j,k,0]$ system, with inhomogeneous term
\[
r_i=\begin{cases} r & i<n \\
r\frac{a_{j+3k-1} - a_{j+3k}}{a_{j+3k-2}}
 & i=n, \end{cases}
\]
we have also not been able to construct sufficiently many commuting integrals. There may be other inhomogeneous terms, i.e. not of the form
$\b=(r,\ldots,r,s)$ for which the number of integrals is sufficient. An example is given in Section \ref{see5}, item 2.
   
\section{Explicit examples of Liouville integrable LV systems} \label{explex}
The number of triples $[j,k,l]$ with $k\geq l$ and $j+k+l=m-1$ is
\[
P(m)=\frac{m^{2}}{4}+\frac{m}{2}+\frac{1-\left(-1\right)^{m}}{8}.
\]
whereas the number of triples $[j,k,l]$ with $k\geq l$, $jk=0$ and $j+k+l=m-1$ is
\[
Q(m)=\frac{m}{2}+\frac{7+(-1)^m}{4}.
\]
The values of $P(m)$ and $Q(m)$ for $m\leq 10$ are given in Table \ref{tabl}.
\begin{table}[h]
\begin{center}
\begin{tabular}{c|cccccccccc}
m & 2 & 3 & 4 & 5 & 6 & 7 & 8 & 9 & 10 \\
\hline
P(m) & 2 & 4 & 6 & 9 & 12 & 16 & 20 & 25 & 30 \\
Q(m) & 2 & 3 & 3 & 4 & 4 & 5 & 5 & 6 & 6\\
\hline
\end{tabular}
\vspace{2mm}

\parbox{12.5cm}{\caption{\label{tabl} The number $P(m)$ of homogeneous $[j,k,l]$ LV systems in dimensions $2m$ and $2m-1$, as well as the number $Q(m)$ of inhomogeneous $2m-1$ dimensional Liouville integrable $[j,k,l]$ LV systems, with inhomogeneous term of the form $\b=(r,\ldots,r,s)$.}}
\end{center}
\end{table}

We illustrate the theory by providing explicitly the families of Lotka-Volterra systems of type $[j,k,l]$, in dimensions $n=3,4,5,6$. These correspond to $n=2m$ or $n=2m-1$ dimensional systems with $j+k+l=m-1=1,2$, i.e. for $m=2$:
\[
[j,k,l]=[1,0,0] \text{ or } [0,1,0],
\]
and for $m=3$:
\[
[j,k,l]=[2,0,0], [1,1,0], [0,1,1]  \text{ or } [0,2,0].
\]
In the following sections, each homogeneous system is understood to have Hamiltonian $H=x_1+\cdots+x_n$. For the inhomogeneous systems the Hamiltonians will be given.
\subsection{The even $[j,k,l]$ LV systems with $j+k+l=1$ ($n=4$)} \label{see4}
\begin{enumerate}
\item The homogeneous [1,0,0] system has matrix
\[
\A=\left[\begin{array}{cccc}
0 & a_{1} & b_{1} & c_{1} 
\\
 -a_{1} & 0 & b_{1} & c_{1} 
\\
 -b_{1} & -b_{1} & 0 & a_{2} 
\\
 -c_{1} & -c_{1} & -a_{2} & 0 
\end{array}\right]
\]
and integral
\[
\left(x_{1}+x_{2}\right)^{a_{2}} x_{3}^{-c_{1}} x_{4}^{b_{1}}.
\]
One can add the vector $\b=(r,r,r,r)$ if one takes $b_1=c_1-a_2$.
\item
The homogeneous [0,1,0] system has matrix
\[
\A=\left[\begin{array}{cccc}
0 & a_{1} & a_{2} & a_{4} 
\\
 -a_{1} & 0 & a_{3} & a_{4} 
\\
 -a_{2} & -a_{3} & 0 & a_{4} 
\\
 -a_{4} & -a_{4} & -a_{4} & 0 
\end{array}\right]
\]
and integral
\[
x_{1}^{-a_{3}} x_{2}^{a_{2}} x_{3}^{-a_{1}} \left(x_{1}+x_{2}+x_{3}\right)^{a_{1}-a_{2}+a_{3}}.
\]
The vector $\b=(r,r,r,r)$ can be added without any further constraints.
\end{enumerate}
\subsection{The odd $[j,k,l]$ LV systems with $j+k+l=1$ ($n=3$)} \label{see3}
\begin{enumerate}
\item The homogeneous [1,0,0] system has matrix
\[
\A=\left[\begin{array}{ccc}
0 & a_{1} & b_{1} 
\\
 -a_{1} & 0 & b_{1} 
\\
 -b_{1} & -b_{1} & 0 
\end{array}\right]
\]
and Casimir $x_{1}^{b_{1}} x_{2}^{-b_{1}} x_{3}^{a_{1}}$. We can add the inhomogeneous term $\b=(r,r,0)$ and then the system has Hamiltonian
\[
H=x_{1}+x_{2}+x_{3}+\frac{r}{a_{1}} \ln\left(\frac{x_{2}}{x_{1}}\right).
\]
\item The homogeneous [0,1,0] system has matrix
\[
\A=\left[\begin{array}{ccc}
0 & a_{1} & a_{2} 
\\
 -a_{1} & 0 & a_{3} 
\\
 -a_{2} & -a_{3} & 0 
\end{array}\right],
\]
and contains the $[1,0,0]$ system as a special case. It has Casimir $x_{1}^{a_{3}} x_{2}^{-a_{2}} x_{3}^{a_{1}}$, and inhomogeneous generalisation $\b=\left(r,r,r\frac{a_{2}-a_{3}}{a_{1}}\right)$ with Hamiltonian
\[
H=x_{1}+x_{2}+x_{3}-r\frac{a_{2}-a_{3}}{a_{3} a_{1}} \ln \left(x_{2}\right)+\frac{r}{a_{3}} \ln\left(x_{3}\right).
\]
\end{enumerate}

\subsection{The even $[j,k,l]$ LV systems with $j+k+l=2$ ($n=6$)} \label{see6}
\begin{enumerate}
\item The homogeneous [2,0,0] LV system has matrix
\[
\A=\left[\begin{array}{cccccc}
0 & a_{1} & \frac{b_{1} c_{2}-b_{2} c_{1}}{a_{3}} & \frac{b_{1} c_{2}-b_{2} c_{1}}{a_{3}} & b_{1} & c_{1} 
\\
 -a_{1} & 0 & \frac{b_{1} c_{2}-b_{2} c_{1}}{a_{3}} & \frac{b_{1} c_{2}-b_{2} c_{1}}{a_{3}} & b_{1} & c_{1} 
\\
 -\frac{b_{1} c_{2}-b_{2} c_{1}}{a_{3}} & -\frac{b_{1} c_{2}-b_{2} c_{1}}{a_{3}} & 0 & a_{2} & b_{2} & c_{2} 
\\
 -\frac{b_{1} c_{2}-b_{2} c_{1}}{a_{3}} & -\frac{b_{1} c_{2}-b_{2} c_{1}}{a_{3}} & -a_{2} & 0 & b_{2} & c_{2} 
\\
 -b_{1} & -b_{1} & -b_{2} & -b_{2} & 0 & a_{3} 
\\
 -c_{1} & -c_{1} & -c_{2} & -c_{2} & -a_{3} & 0 
\end{array}\right]
\]
and commuting integrals
\[
\left(x_{1}+x_{2}\right)^{a_{3}} x_{5}^{-c_{1}} x_{6}^{b_{1}},\quad
\left(x_{3}+x_{4}\right)^{a_{3}} x_{5}^{-c_{2}} x_{6}^{b_{2}}.
\]
To the subsystem with $b_1=c_1-a_3,b_2=c_2-a_3$, which has matrix
\[
\left[\begin{array}{cccccc}
0 & a_{1} & c_{1}-c_{2} & c_{1}-c_{2} & c_{1}-a_{3} & c_{1} 
\\
 -a_{1} & 0 & c_{1}-c_{2} & c_{1}-c_{2} & c_{1}-a_{3} & c_{1} 
\\
 -c_{1}+c_{2} & -c_{1}+c_{2} & 0 & a_{2} & c_{2}-a_{3} & c_{2} 
\\
 -c_{1}+c_{2} & -c_{1}+c_{2} & -a_{2} & 0 & c_{2}-a_{3} & c_{2} 
\\
 -c_{1}+a_{3} & -c_{1}+a_{3} & -c_{2}+a_{3} & -c_{2}+a_{3} & 0 & a_{3} 
\\
 -c_{1} & -c_{1} & -c_{2} & -c_{2} & -a_{3} & 0 
\end{array}\right]
\]
one can add a inhomogeneous term $\b=(r,r,r,r,r,r)$. It has Hamiltonian \[
H=x_{1}+x_{2}+x_{3}+x_{4}+x_{5}+x_{6}+\frac{r}{a_{3}} \ln\left(\frac{x_{6}}{x_{5}}\right),
\] and commuting integrals 
$\left(x_{1}+x_{2}\right)^{a_{3}} x_{5}^{-c_{1}} x_{6}^{c_{1}-a_{3}}$ and $\left(x_{3}+x_{4}\right)^{a_{3}} x_{5}^{-c_{2}} x_{6}^{c_{2}-a_{3}}$.
\item The homogeneous [1,1,0] system has matrix
\[
\A=\left[\begin{array}{cccccc}
0 & a_{1} & b_{1} & b_{1} & b_{1} & c_{1} 
\\
 -a_{1} & 0 & b_{1} & b_{1} & b_{1} & c_{1} 
\\
 -b_{1} & -b_{1} & 0 & a_{2} & a_{3} & a_{5} 
\\
 -b_{1} & -b_{1} & -a_{2} & 0 & a_{4} & a_{5} 
\\
 -b_{1} & -b_{1} & -a_{3} & -a_{4} & 0 & a_{5} 
\\
 -c_{1} & -c_{1} & -a_{5} & -a_{5} & -a_{5} & 0 
\end{array}\right]
\]
and commuting integrals
\[
\left(x_{1}+x_{2}\right)^{a_{5}} \left(x_{3}+x_{4}+x_{5}\right)^{-c_{1}} x_{6}^{b_{1}},\quad 
x_{3}^{-a_{4}} x_{4}^{a_{3}} x_{5}^{-a_{2}} \left(x_{3}+x_{4}+x_{5}\right)^{a_{2}-a_{3}+a_{4}}.
\]
The inhomogeneous [1,1,0] system has matrix
\[
\A=\left[\begin{array}{cccccc}
0 & a_{1} & c_{1}-a_{5} & c_{1}-a_{5} & c_{1}-a_{5} & c_{1} 
\\
 -a_{1} & 0 & c_{1}-a_{5} & c_{1}-a_{5} & c_{1}-a_{5} & c_{1} 
\\
 -c_{1}+a_{5} & -c_{1}+a_{5} & 0 & a_{2} & a_{3} & a_{5} 
\\
 -c_{1}+a_{5} & -c_{1}+a_{5} & -a_{2} & 0 & a_{4} & a_{5} 
\\
 -c_{1}+a_{5} & -c_{1}+a_{5} & -a_{3} & -a_{4} & 0 & a_{5} 
\\
 -c_{1} & -c_{1} & -a_{5} & -a_{5} & -a_{5} & 0 
\end{array}\right],\quad \b=\left[\begin{array}{c}
r\\ r\\ r\\ r\\ r\\ r
\end{array}\right],
\]
Hamiltonian
\[
x_{1}+x_{2}+x_{3}+x_{4}+x_{5}+x_{6}+\frac{b}{a_5}
\left(\frac{-a_{4} \ln\left(x_{3}\right)+a_{3} \ln\left(x_{4}\right)-a_{2} \ln\left(x_{5}\right)}{ a_{2}-a_{3}+a_{4}}+\ln\left(x_{6}\right)\right)
\]
and integrals
\[
\left(x_{1}+x_{2}\right)^{a_{5}} \left(x_{3}+x_{4}+x_{5}\right)^{-c_{1}} x_{6}^{c_{1}-a_{5}}$, $x_{3}^{-a_{4}} x_{4}^{a_{3}} x_{5}^{-a_{2}} \left(x_{3}+x_{4}+x_{5}\right)^{a_{2}-a_{3}+a_{4}}.
\]

\item The homogeneous [0,1,1] system has matrix
\[
\A=\left[\begin{array}{cccccc}
0 & a_{1} & a_{2} & a_{4} & a_{4} & a_{4} 
\\
 -a_{1} & 0 & a_{3} & a_{4} & a_{4} & a_{4} 
\\
 -a_{2} & -a_{3} & 0 & a_{4} & a_{4} & a_{4} 
\\
 -a_{4} & -a_{4} & -a_{4} & 0 & a_{5} & a_{6} 
\\
 -a_{4} & -a_{4} & -a_{4} & -a_{5} & 0 & a_{7} 
\\
 -a_{4} & -a_{4} & -a_{4} & -a_{6} & -a_{7} & 0 
\end{array}\right]
\]
and commuting integrals
\[
x_{1}^{-a_{3}} x_{2}^{a_{2}} x_{3}^{-a_{1}} \left(x_{1}+x_{2}+x_{3}\right)^{a_{1}-a_{2}+a_{3}},\quad 
x_{4}^{-a_{7}} x_{5}^{a_{6}} x_{6}^{-a_{5}} \left(x_{4}+x_{5}+x_{6}\right)^{a_{5}-a_{6}+a_{7}}.
\]
We can allow the inhomogeneous term with $\b=(r,r,r,r,r,r)$, in which case the Hamiltonian is
\begin{align*}
H&=x_{1}+x_{2}+x_{3}+x_{4}+x_{5}+x_{6}+\frac{r}{a_4}
\left(\frac{-a_{3} \ln\left(x_{1}\right)+a_{2} \ln\left(x_{2}\right)-a_{1} \ln\left(x_{3}\right)}{ a_{1}-a_{2}+a_{3}}\right.\\
&\ \ \ \ \ \ \ \ \ \left.+\frac{a_{7} \ln\left(x_{4}\right)-a_{6} \ln\left(x_{5}\right)+a_{5} \ln\left(x_{6}\right)}{ a_{5}-a_{6}+a_{7}} \right).
\end{align*}
\item The homogeneous  [0,2,0] system has matrix
\[
\A=\left[\begin{array}{cccccc}
0 & a_{1} & a_{2} & a_{4} & a_{5} & a_{7} 
\\
 -a_{1} & 0 & a_{3} & a_{4} & a_{5} & a_{7} 
\\
 -a_{2} & -a_{3} & 0 & a_{4} & a_{5} & a_{7} 
\\
 -a_{4} & -a_{4} & -a_{4} & 0 & a_{6} & a_{7} 
\\
 -a_{5} & -a_{5} & -a_{5} & -a_{6} & 0 & a_{7} 
\\
 -a_{7} & -a_{7} & -a_{7} & -a_{7} & -a_{7} & 0 
\end{array}\right]
\]
and commuting integrals
\[
\frac{x_{2}^{a_{2}} \left(x_{1}+x_{2}+x_{3}\right)^{a_{1}-a_{2}+a_{3}}}{x_{1}^{a_{3}}x_{3}^{a_{1}} },\quad 
\frac{ x_{4}^{a_{5}} \left(x_{1}+x_{2}+x_{3}+x_{4}+x_{5}\right)^{a_{4}-a_{5}+a_{6}}
}{\left(x_{1}+x_{2}+x_{3}\right)^{a_{6}}x_{5}^{a_{4}} }.
\]
\end{enumerate}
We can allow the inhomogeneous term with $\b=(r,r,r,r,r,r)$, in which case the Hamiltonian is
\begin{align*}
H&=x_{1}+x_{2}+x_{3}+x_{4}+x_{5}+x_{6}+\frac{r}{a_7}
\left(\frac{a_6(-a_{3} \ln\left(x_{1}\right)+a_{2} \ln\left(x_{2}\right)-a_{1} \ln\left(x_{3}\right))}{ \left(a_{1}-a_{2}+a_{3}\right) \left(a_{4}-a_{5}+a_{6}\right)}\right.\\
&\ \ \ \ \ \ \ \ \ \left.+\frac{a_{5} \ln\left(x_{4}\right)-a_{4} \ln\left(x_{5}\right)}{ a_{4}-a_{5}+a_{6}} + \ln\left(x_{6}\right) \right).
\end{align*}
\subsection{The odd $[j,k,l]$ LV systems with $j+k+l=2$ ($n=5$)} \label{see5}
\begin{enumerate}
\item The homogeneous [2,0,0] system has matrix
\[
\left[\begin{array}{ccccc}
0 & a_{1} & \frac{b_{1} c_{2}-b_{2} c_{1}}{a_{3}} & \frac{b_{1} c_{2}-b_{2} c_{1}}{a_{3}} & b_{1} 
\\
 -a_{1} & 0 & \frac{b_{1} c_{2}-b_{2} c_{1}}{a_{3}} & \frac{b_{1} c_{2}-b_{2} c_{1}}{a_{3}} & b_{1} 
\\
 -\frac{b_{1} c_{2}-b_{2} c_{1}}{a_{3}} & -\frac{b_{1} c_{2}-b_{2} c_{1}}{a_{3}} & 0 & a_{2} & b_{2} 
\\
 -\frac{b_{1} c_{2}-b_{2} c_{1}}{a_{3}} & -\frac{b_{1} c_{2}-b_{2} c_{1}}{a_{3}} & -a_{2} & 0 & b_{2} 
\\
 -b_{1} & -b_{1} & -b_{2} & -b_{2} & 0 
\end{array}\right],
\]
integral $\dfrac{\left(x_{1}+x_{2}\right)^{b_{2}}}{ \left(x_{3}+x_{4}\right)^{b_{1}}}x_{5}^{\frac{c_{2} b_{1}-c_{1} b_{2}}{a_{3}}} $ and Casimir
$
\left(\dfrac{x_{1}}{x_{2}}\right)^{b_{1} a_{2}}
\left(\dfrac{x_{3}}{x_{4}}\right)^{a_{1} b_{2}}
x_{5}^{a_{1} a_{2}}$.

The inhomogeneous [2,0,0] system, with $\b=(r,r,r,r,0)$, has matrix
\[
\A=\left[\begin{array}{ccccc}
0 & a_{1} & c_{2}-c_{1} & c_{2}-c_{1} & a_{3} 
\\
 -a_{1} & 0 & c_{2}-c_{1} & c_{2}-c_{1} & a_{3} 
\\
 c_{1}-c_{2} & c_{1}-c_{2} & 0 & a_{2} & a_{3} 
\\
 c_{1}-c_{2} & c_{1}-c_{2} & -a_{2} & 0 & a_{3} 
\\
 -a_{3} & -a_{3} & -a_{3} & -a_{3} & 0 
\end{array}\right],
\]
Hamiltonian
\[
H=x_{1}+x_{2}+x_{3}+x_4+x_5+\frac{r}{a_{1}} \ln\left(\frac{x_{2}}{x_{1}}\right)+\frac{r}{a_{2}} \ln\left(\frac{x_{4}}{x_{3}}\right),
\]
integrals \[
\left(\frac{x_{1}+x_{2}}{x_{3}+x_{4}}\right)^{1/(c_1-c_2)}
\left(\frac{x_{1}}{x_{2}}\right)^{1/a_1}
\left(\frac{x_{3}}{x_{4}}\right)^{1/a_2}
,\quad
\frac{x_{1}+x_{2}+x_{3}+x_{4}}{x_{1}+x_{2}}
\left(\frac{x_{3}}{x_{4}}\right)^{(c_2-a_1)/a_2}
\]
and Casimir
$
\left(\dfrac{x_{1}}{x_{2}}\right)^{1/a_{1}}
\left(\dfrac{x_{3}}{x_{4}}\right)^{1/a_2}
x_{5}^{1/a_3}$.
\item
The homogeneous [1,1,0] system has matrix
\[
\A=\left[\begin{array}{ccccc}
0 & a_{1} & b_{1} & b_{1} & b_{1} 
\\
 -a_{1} & 0 & b_{1} & b_{1} & b_{1} 
\\
 -b_{1} & -b_{1} & 0 & a_{2} & a_{3} 
\\
 -b_{1} & -b_{1} & -a_{2} & 0 & a_{4} 
\\
 -b_{1} & -b_{1} & -a_{3} & -a_{4} & 0 
\end{array}\right],
\]
integral $\frac{x_{4}^{a_{3}}\left(x_{3}+x_{4}+x_{5}\right)^{a_{2}-a_{3}+a_{4}}}{x_{3}^{a_{4}}x_{5}^{a_{2}} }$ and Casimir $\left(\frac{x_{1}}{x_2}\right)^{\frac{b_{1}}{a_{1}}} \left(\frac{x_{3}^{a_{4}}x_{5}^{a_{2}}}{x_{4}^{a_{3}}}\right)^{1/(a_{2}-a_{3}+a_{4})}$.
One can not add a inhomogeneous term of the form $\b=(r,\ldots,r,s)$. However, one can take $\b=(r,r,0,0,0)$ (without constraints on the parameters) and then the Hamiltonian is
\[
H=x_1+x_2+x_3+x_4+x_5+\frac{r}{a_1}\ln\left(\frac{x_2}{x_1}\right).
\]
\item
The [0,1,1] system has matrix
\[
\A=\left[\begin{array}{ccccc}
0 & a_{1} & a_{2} & a_{4} & a_{4} 
\\
 -a_{1} & 0 & a_{3} & a_{4} & a_{4} 
\\
 -a_{2} & -a_{3} & 0 & a_{4} & a_{4} 
\\
 -a_{4} & -a_{4} & -a_{4} & 0 & a_{5} 
\\
 -a_{4} & -a_{4} & -a_{4} & -a_{5} & 0 
\end{array}\right],
\]
integral $\frac{x_{2}^{a_{2}}\left(x_{1}+x_{2}+x_{3}\right)^{a_{1}-a_{2}+a_{3}} }{x_{1}^{a_{3}}x_{3}^{a_{1}}}$ and Casimir $
\left(\frac{x_{1}^{a_{3}}x_{3}^{a_{1}}}{x_{2}^{a_{2}}}\right)^{a_{5}} \left(\frac{x_{5}}{x_{4}}\right)^{a_{4}(a_{1}-a_{2}+a_{3})} $.
The Hamiltonian $H=x_1+x_2+x_3+x_4+x_5+r\frac{a_5-a_4}{a_4a_5}\ln(x_4)+\frac{r}{a_5}\ln(x_5)$ yields a inhomogeneous term with $\b=\left(r,r,r,r,r\frac{a_4-a_5}{a_4}\right)$.
\item
The [0,2,0] system (which contains the [0,1,1] case) has matrix
\[
\A=\left[\begin{array}{ccccc}
0 & a_{1} & a_{2} & a_{4} & a_{5} 
\\
 -a_{1} & 0 & a_{3} & a_{4} & a_{5} 
\\
 -a_{2} & -a_{3} & 0 & a_{4} & a_{5} 
\\
 -a_{4} & -a_{4} & -a_{4} & 0 & a_{6} 
\\
 -a_{5} & -a_{5} & -a_{5} & -a_{6} & 0 
\end{array}\right],
\]
integral $\frac{x_{2}^{a_{2}}\left(x_{1}+x_{2}+x_{3}\right)^{a_{1}-a_{2}+a_{3}} }{x_{1}^{a_{3}}x_{3}^{a_{1}}}$ and Casimir $
\left(\frac{x_{1}^{a_{3}}x_{3}^{a_{1}}}{x_{2}^{a_{2}}}\right)^{a_6} \left(\frac{x_{5}^{a_{4}}}{x_{4}^{a_{5}}}\right)^{a_{1}-a_{2}+a_{3}} $.
The inhomogeneous term is $\b=\left(r,r,r,r, r\frac{\left(a_{5}-a_{6}\right)}{a_{4}}\right)$ with Hamiltonian
$
H=x_1+x_2+x_3+x_4+x_5+b \frac{a_{6}-a_{5}}{a_{4} a_{6}} \ln\left(x_{4}\right)+\frac{b}{a_{6}}\ln\left(x_{5}\right)$.
\end{enumerate}

\subsection{Examples of Liouville integrable LV systems not of type $[j,k,l]$} \label{Sothex}
There exist Liouville integrable LV systems of high rank with associated hypergraphs which are not in the $[j,k,l]$ family displayed in Figure \ref{F0}. These systems seem to have less parameters, and additional DPs and/or integrals.

In this section, the hyperedges correspond to the DPs of the system, not to associated commuting integrals (which does not even make sense when the matrix $\A$ is singular). We give three examples of $10$-component classes of LV systems.

\begin{enumerate}
\item 
The 10-component LV system with rank 8 matrix 
\begin{equation} \label{SM}
\A=\begin{pmatrix}
0 & a_{1} & a_{1} & a_{1} & a_{2} & a_{2} & a_{2} & a_{3} & a_{3} & a_{3} 
\\
 -a_{1} & 0 & a_{1} & a_{1} & a_{2} & a_{2} & a_{2} & a_{3} & a_{3} & a_{3} 
\\
 -a_{1} & -a_{1} & 0 & 0 & a_{2} & a_{2} & a_{2} & a_{3} & a_{3} & a_{3} 
\\
 -a_{1} & -a_{1} & 0 & 0 & a_{2} & a_{2} & a_{2} & a_{3} & a_{3} & a_{3} 
\\
 -a_{2} & -a_{2} & -a_{2} & -a_{2} & 0 & a_{4} & a_{5} & a_{7} & a_{7} & a_{7} 
\\
 -a_{2} & -a_{2} & -a_{2} & -a_{2} & -a_{4} & 0 & a_{6} & a_{7} & a_{7} & a_{7} 
\\
 -a_{2} & -a_{2} & -a_{2} & -a_{2} & -a_{5} & -a_{6} & 0 & a_{7} & a_{7} & a_{7} 
\\
 -a_{3} & -a_{3} & -a_{3} & -a_{3} & -a_{7} & -a_{7} & -a_{7} & 0 & a_{8} & a_{9} 
\\
 -a_{3} & -a_{3} & -a_{3} & -a_{3} & -a_{7} & -a_{7} & -a_{7} & -a_{8} & 0 & a_{10} 
\\
 -a_{3} & -a_{3} & -a_{3} & -a_{3} & -a_{7} & -a_{7} & -a_{7} & -a_{9} & -a_{10} & 0 
\end{pmatrix}
\end{equation}
admits (apart from the Hamiltonian $H=\sum_{i=1}^{10} x_i$) the following DPs
\[
x_{1}+x_{2}, x_{3}+x_{4}, x_{2}+x_{3}+x_{4}, x_{5}+x_{6}+x_{7}, x_{8}+x_{9}+x_{10}
, x_{1}+x_{2}+x_{3}+x_{4},
\]
that is, it has the associated hypergraph depicted in Fig. \ref{F5}.
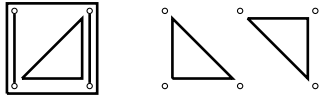
\begin{figure}[h]
\centering
\begin{tikzpicture}[scale=1]
\tikzstyle{nod}= [circle, inner sep=0pt, fill=white, minimum size=2pt, draw]		
\node[nod] (a) at (0,0) {};
\node[nod] (b) at (0,1) {};
\node[nod] (c) at (1,0) {};
\node[nod] (d) at (1,1) {};
\node[nod] (e) at (2,0) {};
\node[nod] (f) at (2,1) {};
\node[nod] (g) at (3,0) {};
\node[nod] (h) at (3,1) {};
\node[nod] (i) at (4,0) {};
\node[nod] (j) at (4,1) {};
\coordinate (a1) at (0.1,0.1) {};
\coordinate (c1) at (.9,0.1) {};
\coordinate (d1) at (.9,.9) {};
\coordinate (e1) at (2.1,0.1) {};
\coordinate (f1) at (2.1,.9) {};
\coordinate (g1) at (2.9,0.1) {};
\coordinate (h1) at (3.1,.9) {};
\coordinate (i1) at (3.9,0.1) {};
\coordinate (j1) at (3.9,0.9) {};
\coordinate (a2) at (-0.1,-0.1) {};
\coordinate (b2) at (-0.1,1.1) {};
\coordinate (c2) at (1.1,-0.1) {};
\coordinate (d2) at (1.1,1.1) {};
\draw[line width=1pt] (a) -- (b);
\draw[line width=1pt] (c) -- (d);
\draw[line width=1pt] (a2) -- (b2) -- (d2) -- (c2) -- (a2);
\draw[line width=1pt] (a1) -- (c1) -- (d1) -- (a1);
\draw[line width=1pt] (e1) -- (f1) -- (g1) -- (e1);
\draw[line width=1pt] (h1) -- (i1) -- (j1) -- (h1);
\end{tikzpicture}
\caption{\label{F5} Hypergraph on $10$ vertices associated to matrix \eqref{SM}.}
\end{figure}
The matrix \eqref{SM} has rank $10-2=8$. The functions
\[
\frac{x_{3}}{x_{4}}, \quad
x_{1}^{-a_{7}} x_{2}^{a_{7}} x_{3}^{-a_{7}} x_{5}^{\frac{a_{3} a_{6}}{a_{4}-a_{5}+a_{6}}} x_{6}^{-\frac{a_{3} a_{5}}{a_{4}-a_{5}+a_{6}}} x_{7}^{\frac{a_{3} a_{4}}{a_{4}-a_{5}+a_{6}}} x_{8}^{-\frac{a_{10} a_{2}}{a_{8}-a_{9}+a_{10}}} x_{9}^{\frac{a_{2} a_{9}}{a_{8}-a_{9}+a_{10}}} x_{10}^{-\frac{a_{8} a_{2}}{a_{8}-a_{9}+a_{10}}}
\]
are Casimirs, and the set of integrals $H$,
\[
\frac{\left(x_{1}+x_{2}+x_{3}+x_{4}\right) x_{2}}{x_{1} x_{3}},\quad 
\frac{\left(x_{5}+x_{6}+x_{7}\right)^{a_{4}-a_{5}+a_{6}} x_{6}^{a_{5}}}{x_{5}^{a_{6}} x_{7}^{a_{4}}},\quad 
\frac{\left(x_{8}+x_{9}+x_{10}\right)^{a_{8}-a_{9}+a_{10}} x_{9}^{a_{9}} }{x_{8}^{a_{10}} x_{10}^{a_{8}}}
\]
is pairwise commuting. Although the hypergraph in Fig. \ref{F5} contains the type $[2,1,1]$ hypergraph, the LV system with matrix \eqref{SM} is not a subsystem of the corresponding LV system.
\item The LV system with rank 8 matrix
\begin{equation} \label{ansa}
\A=\begin{pmatrix}
0 & a_{1} & a_{2} & a_{2} & a_{3} & a_{3} & a_{3} & a_{3} & a_{3} & a_{3} 
\\
 -a_{1} & 0 & a_{2} & a_{2} & a_{3} & a_{3} & a_{3} & a_{3} & a_{3} & a_{3} 
\\
 -a_{2} & -a_{2} & 0 & a_{4} & a_{5} & a_{5} & a_{5} & a_{5} & a_{5} & a_{5} 
\\
 -a_{2} & -a_{2} & -a_{4} & 0 & a_{5} & a_{5} & a_{5} & a_{5} & a_{5} & a_{5} 
\\
 -a_{3} & -a_{3} & -a_{5} & -a_{5} & 0 & a_{6} & a_{7} & a_{7} & a_{7} & a_{7} 
\\
 -a_{3} & -a_{3} & -a_{5} & -a_{5} & -a_{6} & 0 & a_{7} & a_{7} & a_{7} & a_{7} 
\\
 -a_{3} & -a_{3} & -a_{5} & -a_{5} & -a_{7} & -a_{7} & 0 & a_{7} & a_{7} & a_{7} 
\\
 -a_{3} & -a_{3} & -a_{5} & -a_{5} & -a_{7} & -a_{7} & -a_{7} & 0 & a_{8}-a_{9} & a_{8} 
\\
 -a_{3} & -a_{3} & -a_{5} & -a_{5} & -a_{7} & -a_{7} & -a_{7} & -a_{8}+a_{9} & 0 & a_{9} 
\\
 -a_{3} & -a_{3} & -a_{5} & -a_{5} & -a_{7} & -a_{7} & -a_{7} & -a_{8} & -a_{9} & 0 
\end{pmatrix}
\end{equation}
admits the DPs
\[
x_{1}+x_{2},\ x_{3}+x_{4},\ x_{5}+x_{6},\ x_{5}+x_{6}+x_{7},\ 
x_{8}+x_{9}+x_{10},\ x_{7}+x_{8}+x_{9}+x_{10},\ 
x_{5}+x_{6}+x_{7}+x_{8}+x_{9}+x_{10},
\]
see associated hypergraph in Fig. \ref{ahsa}.
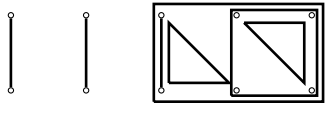
\begin{figure}[h]
\centering
\begin{tikzpicture}[scale=1]
\tikzstyle{nod}= [circle, inner sep=0pt, fill=white, minimum size=2pt, draw]		
\node[nod] (a) at (0,0) {};
\node[nod] (b) at (0,1) {};
\node[nod] (c) at (1,0) {};
\node[nod] (d) at (1,1) {};
\node[nod] (e) at (2,0) {};
\node[nod] (f) at (2,1) {};
\node[nod] (g) at (3,0) {};
\node[nod] (h) at (3,1) {};
\node[nod] (i) at (4,0) {};
\node[nod] (j) at (4,1) {};
\coordinate (e1) at (2.1,0.1) {};
\coordinate (f1) at (2.1,.9) {};
\coordinate (g1) at (2.9,0.1) {};
\coordinate (h1) at (3.1,.9) {};
\coordinate (i1) at (3.9,0.1) {};
\coordinate (j1) at (3.9,0.9) {};
\coordinate (g2) at (2.93,-0.07) {};
\coordinate (h2) at (2.93,1.07) {};
\coordinate (i2) at (4.07,-0.07) {};
\coordinate (j2) at (4.07,1.07) {};
\coordinate (e3) at (1.9,-0.15) {};
\coordinate (f3) at (1.9,1.15) {};
\coordinate (i3) at (4.15,-0.15) {};
\coordinate (j3) at (4.15,1.15) {};
\draw[line width=1pt] (a) -- (b);
\draw[line width=1pt] (c) -- (d);
\draw[line width=1pt] (e) -- (f);
\draw[line width=1pt] (e1) -- (f1) -- (g1) -- (e1);
\draw[line width=1pt] (h1) -- (i1) -- (j1) -- (h1);
\draw[line width=1pt] (g2) -- (h2) -- (j2) -- (i2) -- (g2);
\draw[line width=1pt] (e3) -- (f3) -- (j3) -- (i3) -- (e3);
\end{tikzpicture}
\caption{\label{ahsa} Hypergraph on $10$ vertices associated to matrix \eqref{ansa}.}
\end{figure}
It possesses 8 functionally independent integrals
\begin{align*}
H &= x_{1}+x_{2}+x_{3}+x_{4}+x_{5}+x_{6}+x_{7}+x_{8}+x_{9}+x_{10}\\
C_{1} &= 
\left(\frac{x_{1}}{x_{2}}\right)^{\frac{a_{3}}{a_{1}}}
\left(\frac{x_{3}}{x_{4}}\right)^{\frac{a_{5}}{a_{4}}}
\left(\frac{x_{5}}{x_{6}}\right)^{\frac{a_{7}}{a_{6}}}
x_{7}\left(\frac{x_{10}}{x_{8}}\right)^{\frac{a_{7}}{a_{8}}},\quad 
C_{2} = \frac{x_{9}^{a_{8}} x_{10}^{a_{9}-a_{8}}}{x_{8}^{a_{9}}} \\
R_{1} &=\left(\frac{x_{2}}{x_{1}}\right)^{\frac{a_{3}}{a_{1}}} \left(\frac{x_{4}}{x_{3}}\right)^{\frac{a_{5}}{a_{4}}} \frac{\left(x_{1}+x_{2}\right)^{\frac{a_{5}}{a_{2}}}}{ \left(x_{3}+x_{4}\right)^{\frac{a_{3}}{a_{2}}}} \\
R_{2} &= \left(\frac{x_{1}}{x_{2}}\right)^{\frac{a_{3}}{a_{1}}} \left(\frac{x_{3}}{x_{4}}\right)^{\frac{a_{5}}{a_{4}}} \left(x_{5}+x_{6}+x_{7}+x_{8}+x_{9}+x_{10}\right) \\
R_{3} &= \frac{\left(x_{5}+x_{6}+x_{7}\right)^{a_{6}} x_{6}^{a_{7}}}{x_{5}^{a_{7}}x_{7}^{a_{6}}},\quad 
R_{4} = \frac{\left(x_{7}+x_{8}+x_{9}+x_{10}\right)^{a_{8}} x_{8}^{a_{7}}}{x_{7}^{a_{8}}x_{10}^{a_{7}}} \\
R_{5} &= \left(\frac{x_{1}}{x_{2}}\right)^{\frac{a_{2}}{a_{1}}} \left(x_{3}+x_{4}\right) \left(\frac{x_{6}}{x_{5}}\right)^{\frac{a_{5}}{a_{6}}} \left(\frac{x_{8}+x_{9}+x_{10}}{x_{7}}\right)^{\frac{a_{5}}{a_{7}}},
\end{align*}
where $C_1,C_2$ are Casimirs, and the following brackets vanish
\[
\{R_1,R_2\}=\{R_1,R_3\}=\{R_1,R_4\}=\{R_2,R_3\}=\{R_2,R_4\}=\{R_3,R_5\}=0.
\]
\item
The LV system with nonsingular matrix
\begin{equation} \label{lastone}
\A=\begin{pmatrix}
0 & a_{1} & a_{2} & a_{4} & a_{4} & a_{4} & a_{5} & a_{5} & a_{5} & a_{5} 
\\
 -a_{1} & 0 & a_{3} & a_{4} & a_{4} & a_{4} & a_{5} & a_{5} & a_{5} & a_{5} 
\\
 -a_{2} & -a_{3} & 0 & a_{4} & a_{4} & a_{4} & a_{5} & a_{5} & a_{5} & a_{5} 
\\
 -a_{4} & -a_{4} & -a_{4} & 0 & a_{6} & a_{7} & a_{9} & a_{9} & a_{9} & a_{9} 
\\
 -a_{4} & -a_{4} & -a_{4} & -a_{6} & 0 & a_{8} & a_{9} & a_{9} & a_{9} & a_{9} 
\\
 -a_{4} & -a_{4} & -a_{4} & -a_{7} & -a_{8} & 0 & a_{9} & a_{9} & a_{9} & a_{9} 
\\
 -a_{5} & -a_{5} & -a_{5} & -a_{9} & -a_{9} & -a_{9} & 0 & a_{11} & a_{11} & a_{11} 
\\
 -a_{5} & -a_{5} & -a_{5} & -a_{9} & -a_{9} & -a_{9} & -a_{11} & 0 & a_{10} & a_{11} 
\\
 -a_{5} & -a_{5} & -a_{5} & -a_{9} & -a_{9} & -a_{9} & -a_{11} & -a_{10} & 0 & a_{11} 
\\
 -a_{5} & -a_{5} & -a_{5} & -a_{9} & -a_{9} & -a_{9} & -a_{11} & -a_{11} & -a_{11} & 0
\end{pmatrix}
\end{equation}
admits the DPs
\[
x_{1}+x_{2}+x_{3}, x_{4}+x_{5}+x_{6}, x_{7}+x_{8}+x_{9}
, x_{8}+x_{9}, x_{8}+x_{9}+x_{10}, x_{7}+x_{8}+x_{9}+x_{10},
\]
as depicted in Fig. \ref{F7}.

\begin{figure}[h]
\centering
\begin{tikzpicture}[scale=1]
\tikzstyle{nod}= [circle, inner sep=0pt, fill=white, minimum size=2pt, draw]		
\node[nod] (a) at (0,0) {};
\node[nod] (b) at (0,1) {};
\node[nod] (c) at (1,0) {};
\node[nod] (d) at (1,1) {};
\node[nod] (e) at (2,0) {};
\node[nod] (f) at (2,1) {};
\node[nod] (g) at (3,0) {};
\node[nod] (h) at (3,1) {};
\node[nod] (i) at (4,0) {};
\node[nod] (j) at (4,1) {};
\draw[line width=1pt] (g) -- (j);
\coordinate (a1) at (0.1,0.1) {};
\coordinate (b1) at (0.1,0.9) {};
\coordinate (c1) at (.9,0.1) {};
\draw[line width=1pt] (a1) -- (c1) -- (b1) -- (a1);
\coordinate (d1) at (1.1,0.9) {};
\coordinate (e1) at (1.9,0.1) {};
\coordinate (f1) at (1.9,0.9) {};
\draw[line width=1pt] (f1) -- (e1) -- (d1) -- (f1);
\coordinate (g1) at (3.1,0.2) {};
\coordinate (h1) at (3.1,0.9) {};
\coordinate (j1) at (3.8,0.9) {};
\draw[line width=1pt] (g1) -- (h1) -- (j1) -- (g1);
\coordinate (g2) at (3.2,0.1) {};
\coordinate (j2) at (3.9,0.8) {};
\coordinate (i2) at (3.9,0.1) {};
\draw[line width=1pt] (g2) -- (i2) -- (j2) -- (g2);
\coordinate (g3) at (2.9,-0.1) {};
\coordinate (h3) at (2.9,1.1) {};
\coordinate (i3) at (4.1,-0.1) {};
\coordinate (j3) at (4.1,1.1) {};
\draw[line width=1pt] (g3) -- (h3) -- (j3) -- (i3) -- (g3);
\end{tikzpicture}
\caption{\label{F7} Hypergraph on $10$ vertices associated to matrix \eqref{lastone}.}
\end{figure}
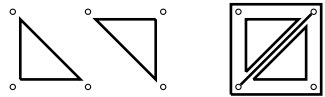

The system has the following set of pairwise commuting integrals
\begin{align*}
T_1&=\left(x_{1}+x_{2}+x_{3}\right)^{-a_{1}+a_{2}-a_{3}} x_{1}^{a_{3}} x_{2}^{-a_{2}} x_{3}^{a_{1}}
,\quad H=x_{1}+x_{2}+\cdots+x_{10},\\
T_2&=\left(x_{4}+x_{5}+x_{6}\right)^{-a_{6}+a_{7}-a_{8}} x_{4}^{a_{8}} x_{5}^{-a_{7}} x_{6}^{a_{6}},\quad
T_3=
\frac{\left(x_{7}+x_{8}+x_{9}\right) \left(x_{8}+x_{9}+x_{10}\right)}{x_{7} x_{10}},\\
T_4&=
\left(\frac{x_1^{a_{3}}x_{3}^{a_{1}}}{x_{2}^{a_{2}}}\right)^{\frac{a_{9}}{a_{1}-a_{2}+a_{3}}} \left(\frac{x_{5}^{a_{7}}}{x_{4}^{a_{8}}x_{6}^{a_{6}}}\right)^{\frac{a_{5}}{a_{6}-a_{7}+a_{8}}}
\left(\frac{x_{7}x_{10}}{x_{8}+x_{9}}\right)^{a_{4}} ,
\end{align*}
and 1 additional integral
$
T_5=x_{8}^{-a_{11}} x_{9}^{a_{11}} x_{10}^{-a_{10}} \left(x_{8}+x_{9}+x_{10}\right)^{a_{10}}$
which commutes with $T_1$ and $T_2$.

The system is not equivalent to a subsystem of the 10-component LV system of type $[2,1,1]$. 
\end{enumerate}

\noindent
We note that the linear transformations that preserve the Lotka-Volterra structure (and induce the equivalence relation introduced in \cite{Hyper}) do not preserve the antisymmetry (nor the skew-symmetrizable property, cf. \cite{FerOli}) of the matrix $\A$. And, they also do not preserve the log-canonical structure of the Poisson bracket.

\section{Concluding remarks} \label{ConclRemks}
The number of different maximal rank Liouville integrable classes of LV systems we provide in this paper grows quadratically in $n$, whereas the number of superintegrable classes of tree-systems \cite{Trees1,Trees2} grows exponentially. The number of parameters in each class, $3n/2-2$, is about half the number of parameters for tree-systems, $3n-2$, and about three quarters of the number of parameters in the Volterra integrable system \eqref{VOLA}, $2n$.
An important difference with Volterra's system is that the rank of his symplectic matrix is minimal (2), and the number of parameters in his linear terms is maximal ($n$). On the contrary, in this paper we have concentrated on symplectic matrices with maximal rank ($n$ resp $n-1$), and we mainly consider one particular form for the linear part $\b=(r,\ldots,r,s)$, i.e. with only 2 parameters.
Although we have identified a large number of integrable classes, we have not exhausted all cases, as the third example in Section \ref{Sothex} shows. The first two examples in Section \ref{Sothex}, of Liouville integrable LV systems with near maximal rank, serve to show that vast territories of untouched ground are still waiting to be charted, namely Liouville integrable LV systems of intermediate rank.
The authors will consider this paper a success, if it encourages the interested reader to join in this search for Liouville integrable systems, in the spirit of Jarmo Hietarinta.

\subsection*{Acknowledgements} We are grateful for the hospitality of Da-jun Zhang and Cheng Zhang of Shanghai University, where this paper was finalised. P.H. van der Kamp was supported by NSFC grant No. 12271334. We thank Ian Marquette for helpful discussions, and anonymous referees for providing historical context, in particular the work of Volterra \cite{Vol2}.

\label{lastpage}
\end{document}